\documentclass[10pt,onecolumn,journal]{IEEEtran}
\usepackage{indentfirst,amsmath}
\usepackage{color}

\usepackage{amsfonts,amsmath,amssymb,amsthm}
\usepackage{latexsym,amscd}
\usepackage{amsbsy}
\usepackage{graphicx,multirow,bm,ulem}
\usepackage{algorithm}
\usepackage{algorithmicx}
\usepackage{algpseudocode}
\usepackage{CJK}
\usepackage{xcolor}
\usepackage{tikz}
\usepackage{caption}
\usepackage{diagbox}
\usepackage{cite}
\usepackage{booktabs}
\usepackage{arydshln}

\newtheorem{Theorem}{Theorem}

\newtheorem{Example}{Example}

\newtheorem{Lemma}{Lemma}

\newcommand{\tabincell}[2]{\begin{tabular}{@{}#1@{}}#2\end{tabular}}
\begin{document}
	\title{A Generic Transformation to Generate MDS Storage Codes with $\delta$-Optimal Access Property
		\author{Yi Liu, Jie Li, and Xiaohu Tang}}
	\maketitle
	\date{}
	\begin{abstract}
		For high-rate maximum distance separable (MDS) codes, most of them are designed to optimally repair a single failed node by connecting all the surviving nodes. However, in practical systems, sometimes not all the surviving nodes are available. To facilitate the practical storage system, a few constructions of $(n,k)$ MDS codes with the property that any single failed node can be optimally repaired by accessing any $d$ surviving nodes have been proposed, where $d\in [k+1:n-1)$. However, all of them either have large sub-packetization levels or are not explicit for all the parameters. To address these issues, we propose a generic transformation that can convert any $(n',k')$ MDS code to another $(n,k)$ MDS code with the optimal repair property and optimal access property for an arbitrary set of two nodes, while the repair efficiency of the remaining $n-2$ nodes can be kept. By recursively applying the generic transformation to a scalar MDS code multiple times, we get an MDS code with  the optimal repair property and the optimal access property for all nodes, which outperforms previous known MDS codes in terms of either the sub-packetization level or the flexibility of the parameters.
	\end{abstract}
	
	
	\section{Introduction}
	
	Building on a large number of unreliable storage nodes, distributed storage systems have important applications in large-scale data center settings, such as Facebook's coded Hadoop, Google Colossus, and Microsoft Azure \cite{Micro}, and also in peer-to-peer storage settings, such as OceanStore \cite{ocean}, Total Recall \cite{total}, and DHash++ \cite{Dhash}. To ensure reliability, it is imperative to store the data in a redundant fashion. In stead of replication, erasure codes have been extensively deployed in distributed storage systems since they can provide higher reliability at the same redundancy level, where MDS codes are one of the most popular erasure codes because they can provide the optimal trade-off between fault tolerance and the storage overhead. When a storage node fails, a self-sustaining distributed storage system should make a repair to maintain the healthy and continuing operation of the overall system.  However, traditional MDS codes such as Reed-Solomon codes \cite{RS codes} only allow a naive repair strategy to repair a failed node, i.e., downloading the amount of the original data to first reconstruct the original file, and then to repair the failed node, this gives rise to  a significantly  large \textit{repair bandwidth}, which is defined as the amount of data downloaded from some surviving nodes to repair the failed node.

	For an $(n,k)$ MDS code with sub-packetization level $N$,  it was proved in \cite{Dimakis} that the repair bandwidth $\gamma(d)$ should satisfy
	\begin{eqnarray}\label{Eqn_bound_on_gamma}
		\gamma(d)\ge \frac{d}{d-k+1}N,
	\end{eqnarray}
	where $d\in [k+1: n)$ is the number of helper nodes.
	If the repair bandwidth of a node attains the lower bound in \eqref{Eqn_bound_on_gamma}, we say that the code has \textit{$\delta$-optimal repair property} for this node, where $\delta=d-k+1\in [2: r]$, i.e., a fraction  $\frac{1}{\delta}$ of the data stored in each of the helper nodes are required to be downloaded. In addition to the optimal repair bandwidth, it is also very desirable if the nodes have the \textit{$\delta$-optimal access property}. That is, the amount of data accessed from the helper nodes meets the lower bound in \eqref{Eqn_bound_on_gamma} during the repair process of a failed node by connecting $d=k+\delta-1$ helper nodes. In fact, the $\delta$-optimal access property implies the $\delta$-optimal repair property, but not vice versa. Recently, it was proved in \cite{tight,tight1} that the sub-packetization level $N$ of $(n,k)$ MDS codes with $\delta$-optimal access property for all nodes is no less than $ \delta^{\lceil\frac{n}{\delta} \rceil}$, where $2\le \delta \le r$, while the codes proposed in \cite{MSR pairty,PCT,Ye_barg_2} with the $r$-optimal access property for all nodes attain this lower bound.
	
	Constructions of MDS codes with the optimal repair bandwidth have been received a lot of attention during the past decade \cite{invariant subspace,Long arxiv,Etzion15,hadamard,Hadamard strategy,Zigzag,extend zigzag,Ye_barg_1,Ye_barg_2,MSR pairty,coupled,kumar,PCT,Hou,Han,Lin,Li ISIT2017,Huang}, while in practical distributed storage systems,  high-rate (i.e., $k>\frac{n}{2}$) is preferred as they require less storage overhead. Among such $(n,k)$ MDS codes, most of them devote to  the case $d=n-1$ (i.e., $\delta=r$) so as to maximally reduce the repair bandwidth since $\gamma(d)=\frac{d}{d-k+1}N$ is a decreasing function of $d$.   Nevertheless, before 2017, the repair of parity nodes are naive for all the high-rate $(n,k)$  MDS codes with $r\ge 3$. By this motivation, Li \textit{et al.} \cite{MSR pairty,Li ISIT2017}  firstly proposed a generic transformation that can convert any nonbinary  $(n,k)$ MDS code into another $(n,k)$ MDS code with $r$-optimal access property for an arbitrary set of $r$ nodes, while keeping the repair/access property  for the remaining $k$ nodes unchanged. As a result,  $(n,k)$ MDS codes with the $r$-optimal access property for all nodes were obtained, with the sub-packetization level being $r^{\lceil\frac{n}{r}\rceil}$. It was shown in \cite{tight1} that the code obtained from the second application of the transformation in \cite{MSR pairty,Li ISIT2017}, the code proposed by Ye. \textit{et al.} \cite{Ye_barg_2} and the one presented by Sasidharan \textit{et al.}\cite{PCT} are essentially equivalent.
	
	Although the generic transformation proposed in \cite{MSR pairty,Li ISIT2017} is a powerful method for building MDS codes with $r$-optimal access property for all nodes, it is not applicable to construct MDS codes with $\delta$-optimal repair property for $2\le \delta<r$.   Recently, a few explicit constructions of high-rate MDS codes with the $\delta$-optimal repair property for $2\le \delta<r$ have been reported in \cite{Ye_barg_1,coupled,kumar}. More precisely,  in \cite{coupled}, Sasidharan \textit{et al.} constructed an $(n=2\delta m,k=2\delta (m-1))$ MDS code of sub-packetization level $\delta \cdot(2\delta)^{m-1}$ by modifying the construction presented in \cite{PCT} (or \cite{Ye_barg_2,MSR pairty}), which is the first-known explicit construction of high-rate MDS code with $\delta$-optimal repair property for all nodes to the best of our knowledge. Through puncturing, one can get the other two MDS codes, i.e., $(n=2\delta m-\Delta_1,k=2\delta (m-1)-\Delta_1)$ MDS code and $(n=2\delta m-\Delta_2,k=2\delta (m-1))$ MDS code, which have the same performance as the original code, where $0\le \Delta_1<2\delta (m-1)$ and $0\le \Delta_2\le \delta$. For convenience, we refer to the code proposed in \cite{coupled} as Sasidharan-Myna-Kumar code.  It is clear that the value $\delta$ of the $(n,k)$ Sasidharan-Myna-Kumar code is required to be larger than $\lceil \frac{r}{2}\rceil$, where $r=n-k$. Later in \cite{Ye_barg_1}, by using diagonal matrices and permutation matrices as the building blocks of the parity-check matrices, Ye and Barg proposed two $(n,k)$ MDS codes with sub-packetization level $\delta^n$, which have the $\delta$-optimal repair and the $\delta$-optimal access property for all nodes, respectively. As a matter of convenience, we refer to the code based on diagonal matrices as Ye-Barg code 1 and the another one as Ye-Barg code 2. Further in \cite{kumar}, an explicit $(n,k)$ MDS code called  Vajha-Babu-Kumar code was presented, which possesses not only  $\delta$-optimal access property for all nodes but also achieves the optimal sub-packetization level $\delta^{\lceil\frac{n}{\delta}\rceil}$, however, this code is implicit for $\delta\ge 5$. Basically, all of the four aforementioned $(n,k)$ MDS codes either have  large sub-packetization levels or are not explicit for any single value of $\delta$ in the range $[2:r]$.  
	
	In this paper, we address the unsolved problem, to build explicit MDS codes with the $\delta$-optimal repair property for all nodes for an arbitrary $\delta \in [2:r]$, and also with a relatively small sub-packetization level. Firstly, we provide a generic transformation that can convert any known MDS code defined in the parity-check matrix form into another MDS code, which makes an arbitrary set of two nodes having the $\delta$-optimal access property, and simultaneously preserves the $\delta$-optimal access property for the remaining nodes if the original ones have. Secondly, recursively applying the transformation to an $(n'=n+\delta\lceil\frac{n}{2}\rceil,k'=k+\delta\lceil\frac{n}{2}\rceil)$ scalar MDS code $\lceil\frac{n}{2}\rceil$ times, we get an explicit construction of high-rate $(n,k)$ MDS code $\mathcal{C}$ of sub-packetization level $\delta^{\lceil\frac{n}{2} \rceil}$ over $\mathbb{F}_q$ with $q\ge n+\delta\lceil\frac{n}{2}\rceil$, which has the $\delta$-optimal access property for all nodes, where $2\le \delta\le r$. A comparison of some key parameters between the four aforementioned codes and new code $\mathcal{C}$ is given in Table \ref{Table comp Ye Kumar}.  It is seen from this comparison that the new MDS code $\mathcal{C}$ obtained from the generic transformation has the following advantages:
	\begin{table*}[htbp]
		\centering
		\caption{A comparison among some $(n,k=n-r)$ MDS code with $\delta$-optimal repair property for all nodes, where $2\le \delta \le r$.}\label{Table comp Ye Kumar}
		\begin{tabular}{|c|c|c|c|c|c|}
			\hline Codes & \tabincell{c}{Sub-packetization\\ level $N$} & Field size $q$ & $\delta$ & Optimal access & Reference\\
			\hline
			\hline \tabincell{c}{Sasidharan-Myna-Kumar\\ code} & $\delta \cdot (2\delta)^{\lceil \frac{n}{2\delta}\rceil}$ & $q\ge n$ & $\lceil \frac{r}{2} \rceil \le \delta \le r$ & No & \cite{coupled}\\
			\hline Ye-Barg code 1 & $\delta^n$ & $q\ge \delta n$ & $2\le \delta \le r$ & No  & \cite{Ye_barg_1}\\
			\hline Ye-Barg code 2 & $\delta^n$ & $q\ge n+1$ & $2\le \delta \le r$ & Yes  & \cite{Ye_barg_1}\\
			\hline Vajha-Babu-Kumar code  & $\delta^{\lceil \frac{n}{\delta}\rceil}$ & $\begin{array}{ll}
				q \ge 6 \lceil \frac{n}{2} \rceil+2, & \delta=2\\
				q \ge 18 \lceil \frac{n}{\delta} \rceil+2, & \delta=3,4\\
			\end{array}$ & $\delta \in \{2,3,4\}$ & Yes  & \cite{kumar}\\
			\hline New Code $\mathcal{C}$ & $\delta^{\lceil \frac{n}{2}\rceil}$ & $q\ge n+\lceil \frac{n}{2}\rceil \delta$ & $2\le \delta \le r$ & Yes & Theorem \ref{Cro the desired code C}\\
			\hline
		\end{tabular}
	\end{table*}
	
	\begin{itemize}
		\item The new $(n,k=n-r)$ code $\mathcal{C}$ has the $\delta$-optimal access property for any single value of $\delta$ with $2\le \delta\le r$.
		\item In contrast to the Sasidharan-Myna-Kumar code, the range of parameter $\delta$ of our new code $\mathcal{C}$ is much broader under the same parameter $n$ and $k$. In addition, Sasidharan-Myna-Kumar code does not have the $\delta$-optimal access property for all nodes but only the $\delta$-optimal repair property.  As a cost, the finite field size and sub-packetization level 
		of the new code $\mathcal{C}$ are larger than those of  Sasidharan-Myna-Kumar code when $\lceil \frac{r}{2} \rceil \le \delta \le r$.
		\item With the same parameters $n$, $k$ and $\delta$, the sub-packetization level of the new $(n,k)$ MDS code $\mathcal{C}$ is much smaller than that of the Ye-Barg codes 1 and 2.  Furthermore, the finite field size of code $\mathcal{C}$ is smaller than that of the Ye-Barg code 1, but a little larger than that of the Ye-Barg code 2.		
		\item Compared with the $(n,k)$ Vajha-Babu-Kumar code which is explicit only for $\delta\in \{2,3,4\}$, the new $(n,k)$ MDS code $\mathcal{C}$ is explicit for all $\delta \in [2:r]$, though its sub-packetization level  is larger for $\delta=3,4$. Besides, when $\delta=2$, the new code $\mathcal{C}$ is built on a smaller finite field and has the same optimal sub-packetization level. 
	\end{itemize}

	\section{Preliminaries}\label{section S2}
	For two non-negative integers $a$ and $b$ with $a<b$, define $[a:b)$ and $[a:b]$ as two ordered sets $\{a,a+1,\cdots,b-1\}$ and $\{a,a+1,\cdots,b\}$, respectively.  Let $\mathbb{F}_q$ be the finite field with $q$ elements where $q$ is a prime power. An $(n,k)$ MDS code encodes a file of size $\mathcal{M}=kN$ into $n$ fragments $\mathbf{f}_0,\mathbf{f}_1,\cdots,\mathbf{f}_{n-1}$, which are stored across $n$ nodes, respectively, where $\mathbf{f}_i=(f_{i,0},f_{i,1},\cdots,f_{i,N-1})^{\top}$ is a column vector of length $N$ over $\mathbb{F}_q$ and $\top$ denotes the transpose operator. In this paper, the MDS codes are assumed to be defined by the following parity-check form:
	\begin{equation}\label{Eqn Code C}
		\underbrace{\left(
			\begin{array}{cccc}
				A_{0,0} & A_{0,1} & \cdots & A_{0,n-1}\\
				A_{1,0} & A_{1,1} & \cdots & A_{1,n-1}\\
				\vdots & \vdots &\ddots & \vdots\\
				A_{r-1,0}&A_{r-1,1}&\cdots&A_{r-1,n-1}\\
			\end{array}
			\right)}_{block~ matrix~ A}
		\left(
		\begin{array}{c}
			\mathbf{f}_0\\\mathbf{f}_1\\\vdots\\\mathbf{f}_{n-1}\\
		\end{array}
		\right)=\mathbf{0}_{rN}
	\end{equation}
	where $r=n-k$, $\mathbf{0}_{N}$ (resp. $\mathbf{0}_{N\times M})$ denotes the zero column vector of length $N$ (resp. the $N\times M$ zero matrix), and will be abbreviated as $\mathbf{0}$ in the sequel if the length (resp. dimensions) is clear.
	
	Note that for each $t\in [0:r)$, $\sum\limits_{i=0}^{n-1}A_{t,i}\mathbf{f}_i=\mathbf{0}$ contains $N$ equations, for convenience, we say that $\sum\limits_{i=0}^{n-1}A_{t,i}\mathbf{f}_i=\mathbf{0}$ is the $t$-th \textit{parity-check group} (PCG), where $A_{t,i}$ is an $N\times N$ matrix  over $\mathbb{F}_q$, called the \textit{parity-check matrix} of node $i$ for the $t$-th PCG. Moreover, for a given $t\in [0:r)$ and $i\in [0:n)$, we call the data $A_{t,i}\mathbf{f}_i$ as the \textit{PCG-data} of node $i$ in the $t$-th PCG. 	
	%	\begin{Remark}
	%		Actually, the code based on generation matrix (systematic code) also can be regard as a special code defined by \eqref{Eqn Code C}, whose parity-check matrix $A_{t,i_s}=I_N$ if $t=s$ and $A_{t,i_s}=\mathbf{0}$ otherwise for a set $\{i_1,i_2,\cdots,i_r\}\subset [0:n)$, where $I_a$ denotes an identity matrix of order $a$.
	%	\end{Remark}
	
	%	\begin{Remark}
	%		To the best of our knowledge, for any $i\in [0:n)$, there is at least one value $t\in [0:r)$ such that $A_{t,i}$ is nonsingular for all the known codes defined by \eqref{Eqn Code C}. In this paper, the codes defined by \eqref{Eqn Code C} are also based on this fact.
	%	\end{Remark}
	
	\subsection{MDS Property}\label{section sec2-1}
	
	An $(n,k)$ code is said to have the MDS property if the original file can be reconstructed by connecting any $k$ out of the $n$ nodes.
	That is, the data stored in any set of  $r=n-k$ nodes can be obtained by the remaining $k$ nodes. By \eqref{Eqn Code C}, we have
	\begin{eqnarray*}
		\underbrace{\left(
			\begin{array}{cccc}
				A_{0,i_1} & A_{0,i_2} & \cdots & A_{0,i_r}\\
				A_{1,i_1} & A_{1,i_2} & \cdots & A_{1,i_r}\\
				\vdots & \vdots &\ddots & \vdots\\
				A_{r-1,i_1}&A_{r-1,i_2}&\cdots&A_{r-1,i_r}\\
			\end{array}
			\right)}_{block~ matrix~ A_r}\left(
		\begin{array}{c}
			\mathbf{f}_{i_1}\\\mathbf{f}_{i_2}\\\vdots\\\mathbf{f}_{i_r}\\
		\end{array}
		\right)=-\left(
		\begin{array}{cccc}
			A_{0,j_1} & A_{0,j_2} & \cdots & A_{0,j_k}\\
			A_{1,j_1} & A_{1,j_2} & \cdots & A_{1,j_k}\\
			\vdots & \vdots &\ddots & \vdots\\
			A_{r-1,j_1}&A_{r-1,j_2}&\cdots&A_{r-1,j_k}\\
		\end{array}
		\right)\left(
		\begin{array}{c}
			\mathbf{f}_{j_1}\\\mathbf{f}_{j_2}\\\vdots\\\mathbf{f}_{j_k}\\
		\end{array}
		\right)
	\end{eqnarray*}
	where $\{i_1,\cdots,i_r\}\subseteq [0:n)$ and $\{j_1,\cdots,j_k\}= [0:n)\backslash\{i_1,\cdots,i_r\}$, which implies that the code has the MDS property if any $r\times r$ sub-block matrix of $A$ is nonsingular. 
	
	\subsection{Optimal repair property}
	According to \eqref{Eqn_bound_on_gamma}, when repairing a failed  node $i\in [0:n)$ of an $(n,k)$ MDS code by connecting any $d=k+\delta-1$ ($2\le \delta \le r$) surviving  nodes, the $\delta$-optimal repair property demands to download $\frac{N}{\delta}$ symbols from each helper node $j\in \mathcal{H}$, where $\mathcal{H}$ denotes the set of indices of the $d$ helper nodes. In fact, the data downloaded from helper node $j$ can be represented by $R_{i,j,\delta}^{\mathcal{H}}\mathbf{f}_j$, where $R_{i,j,\delta}^{\mathcal{H}}$ is an $\frac{N}{\delta}\times N$ matrix of full rank. Like most work in the literature \cite{invariant subspace,MSR pairty,hadamard,Etzion15,PCT,Hadamard strategy,Zigzag,Ye_barg_1,Ye_barg_2,Long arxiv}, throughout this paper,  to simplify the repair strategy,  we always assume $R_{i,j,\delta}^{\mathcal{H}}=R_{i,\delta}$ for any set $\mathcal{H}$,  all $0 \le i < n$ and $j \in \mathcal{H}$, where the $\frac{N}{\delta}\times N$  matrix $R_{i,\delta}$ is called the \textit{$\delta$-repair matrix} of node $i$. In addition,	the code is preferred to have the  $\delta$-optimal access property, i.e., when repairing a failed node $i\in [0:n)$, the amount of accessed data attains the lower bound in \eqref{Eqn_bound_on_gamma}. Clearly, node $i$ has the $\delta$-optimal access property if the repair matrix $R_{i,\delta}$  satisfies that each row has only one nonzero element.
	
	Obviously, we need at least $N$ linear independent equations out of those $rN$ parity-check equations in  \eqref{Eqn Code C}  to recover the  $N$ unknowns of $\mathbf{f}_i$, $i\in [0:n)$. In this paper, similar to that in \cite{Li_near}, for convenience, we only consider those liner combinations from the same parity-check group. Precisely, for any $t\in [0:r)$, we acquire $\frac{N}{\delta}$  linear independent equations from the $t$-th PCG of \eqref{Eqn Code C} by multiplying it with an $\frac{N}{\delta} \times N$ matrix $S_{i,\delta}$ of rank $\frac{N}{\delta}$, i.e.,
	\begin{eqnarray}\label{Eqn linear equations for t}
		S_{i,\delta} A_{t,i}\mathbf{f}_i+\sum_{j\in \mathcal{D}}S_{i,\delta} A_{t,j}\mathbf{f}_j+\sum_{j\in \mathcal{H}}S_{i,\delta} A_{t,j}\mathbf{f}_j=\mathbf{0},\,\,t\in [0:r),
	\end{eqnarray}
	where $S_{i,\delta}$ is called the \textit{$\delta$-select matrix} of node $i$ and  $\mathcal{D}=[0:n)\backslash(\mathcal{H}\cup \{i\})$ is the set of  indices of the $r-\delta$ nodes which are not connected. In particular, $\mathcal{D}=\emptyset$ if  $\delta=r$.
	
	By \eqref{Eqn linear equations for t}, the following linear system of equations (LSE) are available
	\begin{eqnarray}\label{Eqn linear equations}
		\underbrace{\left(\begin{array}{c}
				S_{i,\delta} A_{0,i}\\
				S_{i,\delta} A_{1.i}\\
				\vdots\\
				S_{i,\delta} A_{r-1,i}
			\end{array}\right)\mathbf{f}_i}_{\mathrm{useful~ data}}+\sum_{j\in \mathcal{D}}\underbrace{\left(\begin{array}{c}
				S_{i,\delta} A_{0,j}\\
				S_{i,\delta} A_{1.j}\\
				\vdots\\
				S_{i,\delta} A_{r-1,j}
			\end{array}\right)\mathbf{f}_j}_{\mathrm{intermediate~ data}}
		+\sum_{j\in \mathcal{H}}\underbrace{\left(\begin{array}{c}
				S_{i,\delta} A_{0,j}\\
				S_{i,\delta} A_{1.j}\\
				\vdots\\
				S_{i,\delta} A_{r-1,j}
			\end{array}\right)\mathbf{f}_j}_{\mathrm{interference~ by}~\mathbf{f}_{j}}
		=\mathbf{0}.
	\end{eqnarray}
	Therefore, the optimal repair property indicates that the interference terms caused by $\mathbf{f}_{j}$ can be cancelled by the downloaded data $R_{i,\delta} \mathbf{f}_{j}$ from node $j\in \mathcal{H}$, i.e.,
	\begin{eqnarray}\label{repair node requirement2}
		\mathrm{Rank} (\left(
		\begin{array}{c}
			R_{i,\delta} \\
			S_{i,\delta} A_{t,j} \\
		\end{array}
		\right)) =\frac{N}{\delta} \textrm{~for~all~} j\in \mathcal{H} \textrm{~and~}t\in [0:r),
	\end{eqnarray}
	which implies that there exists an $\frac{N}{\delta}\times \frac{N}{\delta}$ matrix $\tilde{A}_{t,j,i,\delta}$ such that
	\begin{eqnarray}\label{repair node requirement3}
		S_{i,\delta}A_{t,j}=\tilde{A}_{t,j,i,\delta}R_{i,\delta}  \textrm{~for~}  j\in [0:n)\backslash \{i\} \textrm{~and~}t\in [0:r)
	\end{eqnarray}
	since $\mathcal{H}$ is arbitrary. Let $\mathcal{D}=\{j_0,j_1,\cdots,j_{r-\delta-1}\}$, by substituting  \eqref{repair node requirement3} into \eqref{Eqn linear equations}, together with the data $R_{i,\delta}\mathbf{f}_j$ downloaded from each helper node $j\in \mathcal{H}$,  \eqref{Eqn linear equations} can be reduced to
	\begin{eqnarray}\label{Eqn linear equations eq 1}
		\left(\begin{array}{cccc}
			S_{i,\delta} A_{0,i} & \tilde{A}_{0,j_0,i,\delta} & \cdots & \tilde{A}_{0,j_{r-\delta-1},i,\delta}\\
			S_{i,\delta} A_{1,i}  & \tilde{A}_{1,j_0,i,\delta} & \cdots & \tilde{A}_{1,j_{r-\delta-1},i,\delta}\\
			\vdots & \vdots & \vdots & \vdots\\
			S_{i,\delta} A_{r-1,i}  & \tilde{A}_{r-1,j_0,i,\delta} & \cdots & \tilde{A}_{r-1,j_{r-\delta-1},i,\delta}
		\end{array}\right)\left(\begin{array}{c}
			\mathbf{f}_i\\
			R_{i,\delta}\mathbf{f}_{j_0}\\
			\vdots\\
			R_{i,\delta}\mathbf{f}_{j_{r-\delta-1}}
		\end{array}\right)=*,
	\end{eqnarray}
	where the symbol $*$ indicates a known vector in an LSE throughout this paper. It is clear that there are $N+\frac{N}{\delta}|\mathcal{D}|=N+\frac{N}{\delta}(r-\delta)=\frac{rN}{\delta}$ unknown variables with $\frac{rN}{\delta}$ equations in \eqref{Eqn linear equations}, so is that in \eqref{Eqn linear equations eq 1}. Then we have the following result.
	\begin{Lemma}\label{lem the requirement of obtaining fi}
		For given $\delta\in [2:r]$ and $i\in [0:n)$, if node $i$ has the $\delta$-optimal repair property, then the $\frac{rN}{\delta}\times \frac{rN}{\delta}$ coefficient matrix of \eqref{Eqn linear equations eq 1} is nonsingular for any $r-\delta$ subset $\{j_0,j_1,\cdots,j_{r-\delta-1}\}$ of $[0:n)\backslash\{i\}$.
	\end{Lemma}

	\section{A Generic Transformation for MDS Codes}\label{sec generic}
	In this section, we propose a method that can transform an $(n',k')$ MDS code to a new $(n=n'-\delta, k=k'-\delta)$ MDS code with the $\delta$-optimal access property for an arbitrary set of two nodes, while maintaining the $\delta$-optimal access property  of  the remaining  nodes  if  the original ones have. Specially, the two nodes which we wish to endow with the $\delta$-optimal access property in the desired code are called the \textit{goal nodes}, while the other nodes are named the \textit{remainder nodes}. Without loss of generality, we always assume that the first two nodes are the goal nodes unless
	otherwise stated.

	\subsection{The Generic Transformation} 
	Let $n',k'$ be two positive integers with $r=n'-k'$ and $2\le \delta\le r$.
	Initialize  $\mathcal{C}_0$ as an $(n',k')$ MDS code with sub-packetization level $N$. Denote by $(A_{t,i})_{t\in[0:r), i\in[0:n')}$  its parity-check matrix. 
	
	The generic transformation is then carried out through the following three steps.\\
	\textbf{\textit{Step 1. Obtain an $(n=n'-\delta, k=k'-\delta)$ MDS code $\mathcal{C}_1$ by deleting the last $\delta$ nodes of base code $\mathcal{C}_0$.}}\\
	\textbf{\textit{Step 2: Generate an intermediate $(n,k)$ MDS code $\mathcal{C}_2$ by space sharing $\delta$ instances of code $\mathcal{C}_1$.}}
	
	Let $(\mathbf{g}_0^{(a)},\mathbf{g}_1^{(a)},\mathbf{f}_0^{(a)},\cdots,\mathbf{f}_{k-2}^{(a)},\mathbf{f}_{k-1}'^{(a)},\cdots,\mathbf{f}_{n-2}'^{(a)})$ and $(\mathbf{g}_0^{(b)},\mathbf{g}_1^{(b)},\mathbf{f}_0^{(b)},\cdots,\mathbf{f}_{k-2}^{(b)},\mathbf{f}_{k-1}^{(b)},\cdots,\mathbf{f}_{n-2}^{(b)})$ respectively be the stored data of instances $a\in [0:2)$ and $b\in [2:\delta)$ of code $\mathcal{C}_1$. Let $\mathbf{G}_{t,i}$, $\mathbf{F}_{t,j}$ denote the PCG-data of the respective nodes, i.e.,
	\begin{eqnarray*}
		\mathbf{G}_{t,i}=\left(\begin{array}{c}
			\mathbf{G}_{t,i}^{(0)}\\
			\vdots\\
			\mathbf{G}_{t,i}^{(\delta-1)}
		\end{array}\right),\,\,\mathbf{F}_{t,j}=\left(\begin{array}{c}
			\mathbf{F}_{t,j}^{(0)}\\
			\vdots\\
			\mathbf{F}_{t,j}^{(\delta-1)}
		\end{array}\right)\textrm{ for } t\in [0:r),
	\end{eqnarray*}
	where  $\mathbf{G}_{t,i}^{(a)}$ and $\mathbf{F}_{t,j}^{(a)}$ respectively denote the PCG-data of goal node $i\in[0:2)$ and remainder node $j\in [0:n-2)$ in instance $a\in [0:\delta)$, and are column vectors of length $N$ defined by
	\begin{eqnarray*}
		\mathbf{G}_{t,i}^{(a)}=A_{t,i}\mathbf{g}_i^{(a)}, & \mathbf{F}_{t,j}^{(a)}=\left\{\begin{array}{ll}
			A_{t,j+2}\mathbf{f}_j'^{(a)}, & \textrm{if }j\in [k-2:n-2) \textrm{ and } a\in [0:2),\\
			A_{t,j+2}\mathbf{f}_j^{(a)}, & \textrm{otherwise},
		\end{array}\right. 
	\end{eqnarray*}
	\textbf{\textit{Step 3: Construct the desired code $\mathcal{C}_3$  by modifying the PCG-data in the goal nodes of $\mathcal{C}_2$.}}
	
	The change of PCG-data leads to new parity-check equations, which means that the data in some $r$ nodes will be changed. By convention, we assume that the data $\mathbf{f}_j'^{(l)}$ stored at remainder node $j\in [k-2:n-2)$ of instance $l\in [0:2)$ is modified to $\mathbf{f}_j^{(l)}$ and the data at the other nodes are unchanged.	
	
	Figure \ref{fig t-PCG of (n,k) Code C3} depicts the structure of the new code $\mathcal{C}_3$. Denote $\mathbf{G}_{t,i}'$ and $\mathbf{F}_{t,j}'$ the PCG-data stored at goal node $i$ and remainder node $j$ of code $\mathcal{C}_3$, where $i\in [0:2)$ and $j\in [k-2:n-2)$. For convenience, we write $\mathbf{G}_{t,i}'$ and $\mathbf{F}_{t,j}'$ as
	\begin{eqnarray*}
		\mathbf{G}_{t,i}'=\left(\begin{array}{c}
			\mathbf{G}_{t,j}'^{(0)}\\
			\vdots\\
			\mathbf{G}_{t,i}'^{(\delta-1)}
		\end{array}\right),i\in [0:2),~\mathbf{F}_{t,j}'=\left(\begin{array}{c}
			\mathbf{F}_{t,j}'^{(0)}\\
			\vdots\\
			\mathbf{F}_{t,j}'^{(\delta-1)}
		\end{array}\right),j\in [k-2:n-2)
	\end{eqnarray*}
	where $\mathbf{G}_{t,i}'^{(a)}$ and $\mathbf{F}_{t,j}'^{(a)}~(a\in [0:\delta))$ are column vectors of length $N$ defined by
	\begin{eqnarray*}
		\mathbf{G}_{t,i}'^{(a)}=\left\{\begin{array}{ll}
			\mathbf{G}_{t,i}^{(a)}-\sum\limits_{u=0,u\ne i}^{\delta-1}A_{t,n+u}\mathbf{g}_i^{(u)}, & \textrm{if~} i=a, \\
			A_{t,n+a}\mathbf{g}_i^{(a)}, & \textrm{if } 0\le i \ne a<2,\\
			\mathbf{G}_{t,i}^{(a)}, & \textrm{otherwise,}
		\end{array}\right.
	\end{eqnarray*}
	and
	\begin{eqnarray*}
		\mathbf{F}_{t,j}'^{(a)}=\left\{\begin{array}{ll}
			A_{t,j+2}\mathbf{f}_j^{(a)}, & \textrm{if } a\in [0:2),\\
			\mathbf{F}_{t,j}^{(a)}, & \textrm{otherwise}.
		\end{array}\right.
	\end{eqnarray*}
	
	\begin{figure}[!htbp]
		\tikzset{
			source1/.style={draw,thick,rounded corners,align=center,inner sep=.1cm,minimum width=1.5cm,minimum height=1.4cm},
			source2/.style={draw,thick,rounded corners,align=center,inner sep=.1cm,minimum width=5.6cm,minimum height=1.4cm},
			global scale/.style={scale=#1,every node/.append style={scale=#1}}
		}
		\centering
		\begin{tikzpicture}[global scale=0.85]
			%=======================================node number=======================================================%
			\node [source2] (GN 0) at (0,0.1)  {\small PCG-data of GN 0 ($\mathbf{G}_{t,0}'$)};
			%\node at (3.15,0) {$+$};
			\node [source2] (GN 2) at (6.3,0.1)  {\small PCG-data of GN 1 ($\mathbf{G}_{t,1}'$)};
			%\node at (9.45,0) {$+$};
			\node [source1] (RN 0) at (10.8,0.1)  {\small PCG-data of\\ \small RN 0};
			%\node at (12.1,0.1) {$+$};
			\node  at (12.55,0.1)  {$\cdots$};
			%\node at (13.0,0.1) {$+$};
			\node [source1] at (14.2,0.1) {\small PCG-data of\\ \small RN $n-3$};
			%====================================PCG-data of first row===================================================%
			\node at (0,-1.4)  {\small$A_{t,0}\mathbf{g}_0^{(0)}\hspace{-1mm}-\hspace{-1mm}\dashuline{A_{t,n+1}\mathbf{g}_0^{(1)}}\hspace{-1mm}-\hspace{-1mm}\uline{\sum\limits_{u=2}^{\delta-1}A_{t,n+u}\mathbf{g}_0^{(u)}}$};
			\node at (3.15,-1.4) {$+$};
			\node at (6.3,-1.4)  {\small$A_{t,n}\mathbf{g}_1^{(0)}$};
			\node at (9.45,-1.4) {$+$};
			\node at (10.8,-1.4)  {\small$A_{t,2}\mathbf{f}_0^{(0)}$};
			\node at (12.1,-1.4) {$+$};
			\node at (12.55,-1.4) {$\cdots$};
			\node at (13.0,-1.4) {$+$};
			\node at (14.2,-1.4) {\small$A_{t,n-1}\mathbf{f}_{n-3}^{(0)}$};
			\node at (15.7,-1.4)  {$=\mathbf{0}$};
			%====================================PCG-data of second row===================================================%
			\node at (0,-2.3)  {\small$A_{t,n+1}\mathbf{g}_0^{(1)}$};
			\node at (3.15,-2.3) {$+$};
			\node at (6.3,-2.3)  {\small$A_{t,1}\mathbf{g}_1^{(1)}-\hspace{-1mm}\dashuline{A_{t,n}\mathbf{g}_1^{(0)}}\hspace{-1mm}-\hspace{-1mm}\uline{\sum\limits_{u=2}^{\delta-1}A_{t,n+u}\mathbf{g}_1^{(u)}}$};
			\node at (9.45,-2.3) {$+$};
			\node at (10.8,-2.3)  {\small$A_{t,2}\mathbf{f}_0^{(1)}$};
			\node at (12.1,-2.3) {$+$};
			\node at (12.55,-2.3) {$\cdots$};
			\node at (13.0,-2.3) {$+$};
			\node at (14.2,-2.3) {\small$A_{t,n-1}\mathbf{f}_{n-3}^{(1)}$};
			\node at (15.7,-2.3)  {$=\mathbf{0}$};
			%====================================PCG-data of third row===================================================%
			\node at (0,-3.2)  {\small$A_{t,0}\mathbf{g}_0^{(2)}$};
			\node at (3.15,-3.2) {$+$};
			\node at (6.3,-3.2)  {\small$A_{t,1}\mathbf{g}_1^{(2)}$};
			\node at (9.45,-3.2) {$+$};
			\node at (10.8,-3.2)  {\small$A_{t,2}\mathbf{f}_0^{(2)}$};
			\node at (12.1,-3.2) {$+$};
			\node at (12.55,-3.2) {$\cdots$};
			\node at (13.0,-3.2) {$+$};
			\node at (14.2,-3.2) {\small$A_{t,n-1}\mathbf{f}_{n-3}^{(2)}$};
			\node at (15.7,-3.2)  {$=\mathbf{0}$};
			%====================================PCG-data of fourth row===================================================%
			\node at (0,-3.9)  {$\vdots$};
			\node at (3.15,-3.9) {$\vdots$};
			\node at (6.3,-3.9)  {$\vdots$};
			\node at (9.45,-3.9) {$\vdots$};
			\node at (10.8,-3.9)  {$\vdots$};
			\node at (12.1,-3.9) {$\vdots$};
			\node at (12.55,-4.05) {$\cdots$};
			\node at (13.0,-3.9) {$\vdots$};
			\node at (14.2,-3.9) {$\vdots$};
			\node at (15.55,-3.9)  {$\vdots$};
			%====================================PCG-data of fifth row===================================================%
			\node at (0,-4.6)  {\small$A_{t,0}\mathbf{g}_0^{(\delta-1)}$};
			\node at (3.15,-4.6) {$+$};
			\node at (6.3,-4.6)  {\small$A_{t,1}\mathbf{g}_1^{(\delta-1)}$};
			\node at (9.45,-4.6) {$+$};
			\node at (10.8,-4.6)  {\small$A_{t,2}\mathbf{f}_0^{(\delta-1)}$};
			\node at (12.1,-4.6) {$+$};
			\node at (12.55,-4.6) {$\cdots$};
			\node at (13.0,-4.6) {$+$};
			\node at (14.2,-4.6) {\small$A_{t,n-1}\mathbf{f}_{n-3}^{(\delta-1)}$};
			\node at (15.7,-4.6)  {$=\mathbf{0}$};
			
			%=======================================Group===========================================%
			\draw[dashed] (-2.8,-0.85) rectangle (2.8,-5.0);
			\draw[dashed] (3.5,-0.85) rectangle (9.1,-5.0);
			\draw[dashed] (9.8,-0.85) rectangle (11.8,-5.0);
			\draw[dashed] (13.2,-0.85) rectangle (15.2,-5.0);
		\end{tikzpicture}
		\captionsetup{justification=centering}
		\caption{The $t$-th PCG of new $(n,k)$ MDS code $\mathcal{C}_3$ for $t\in [0:r)$, where GN and RN respectively denote goal node and remainder node.}
		\label{fig t-PCG of (n,k) Code C3}
	\end{figure}

	Let us first look at an example of the transformation.	
	\begin{Example}
		%	Our ultimate goal is to obtain a $(12,8)$ MDS code with $3$-optimal access property for the first two nodes. 
		We choose a  known $(15,11)$ MDS code $\mathcal{C}_0$ with sub-packetization level $N$ as the base code, whose parity-check matrix is $(A_{t,i})_{t\in [0: 4), i \in [0:15)}$. Let $\delta=3$, applying the generic transformation, we can get the desired $(12,8)$ MDS code $\mathcal{C}_3$ with sub-packetization level $3N$, whose $t$-th PCG is given in Figure \ref{fig t-PCG of (12,8) Code C2}.
		\begin{figure}[!htbp]
			\tikzset{
				source1/.style={draw,thick,rounded corners,align=center,inner sep=.1cm,minimum width=1.5cm,minimum height=1.4cm},
				source2/.style={draw,thick,rounded corners,align=center,inner sep=.1cm,minimum width=5.0cm,minimum height=1.4cm},
				global scale/.style={scale=#1,every node/.append style={scale=#1}}
			}
			\centering
			\begin{tikzpicture}[global scale=0.85]			
				%=======================================node number=======================================================%
				\node [source2] (GN 0) at (0,0)  {PCG-data of GN 0};
				\node [source2] (GN 2) at (5.6,0)  {PCG-data of GN 1};
				\node [source1] (RN 0) at (9.6,0)  {PCG-data\\ of RN 0};
				%\node at (9.7,0) {$+$};
				\node [source1] (RN 0) at (12,0)  {PCG-data\\ of RN 1};
				%\node at (13.1,0) {$+$};
				\node at (13.7,0) {$\cdots$};
				%\node at (13.9,0) {$+$};
				\node [source1] (RN 6) at (15.3,0)  {PCG-data\\ of RN 9};
				%====================================PCG-data of first row===================================================%
				\node at (0,-1.3)  {$A_{t,0}\mathbf{g}_0^{(0)}\hspace{-1mm}-\hspace{-1mm}A_{t,13}\mathbf{g}_0^{(1)}\hspace{-1mm}-\hspace{-1mm}A_{t,14}\mathbf{g}_0^{(2)}$};
				\node at (2.8,-1.3) {$+$};
				\node at (5.6,-1.3)  {$A_{t,12}\mathbf{g}_1^{(0)}$};
				\node at (8.4,-1.3) {$+$};
				\node at (9.6,-1.3)  {$A_{t,2}\mathbf{f}_0^{(0)}$};
				\node at (10.8,-1.3) {$+$};
				\node at (12.0,-1.3) {$A_{t,3}\mathbf{f}_1^{(0)}$};
				\node at (13.2,-1.3) {$+$};
				\node at (13.7,-1.3) {$\cdots$};
				\node at (14.2,-1.3) {$+$};
				\node at (15.3,-1.3) {$A_{t,11}\mathbf{f}_9^{(0)}$};
				\node at (16.7,-1.3)  {$=\mathbf{0}$};
				%====================================PCG-data of second row===================================================%
				\node at (0,-2.0)  {$A_{t,13}\mathbf{g}_0^{(1)}$};
				\node at (2.8,-2.0) {$+$};
				\node at (5.6,-2.0)  {$A_{t,1}\mathbf{g}_1^{(1)}\hspace{-1mm}-\hspace{-1mm}A_{t,12}\mathbf{g}_1^{(0)}\hspace{-1mm}-\hspace{-1mm}A_{t,14}\mathbf{g}_1^{(2)}$};
				\node at (8.4,-2.0) {$+$};
				\node at (9.6,-2.0)  {$A_{t,2}\mathbf{f}_0^{(1)}$};
				\node at (10.8,-2.0) {$+$};
				\node at (12.0,-2.0) {$A_{t,3}\mathbf{f}_1^{(1)}$};
				\node at (13.2,-2.0) {$+$};
				\node at (13.7,-2.0) {$\cdots$};
				\node at (14.2,-2.0) {$+$};
				\node at (15.3,-2.0) {$A_{t,11}\mathbf{f}_9^{(1)}$};
				\node at (16.7,-2.0)  {$=\mathbf{0}$};
				%====================================PCG-data of third row===================================================%
				\node at (0,-2.7)  {$A_{t,0}\mathbf{g}_0^{(2)}$};
				\node at (2.8,-2.7) {$+$};
				\node at (5.6,-2.7)  {$A_{t,1}\mathbf{g}_1^{(2)}$};
				\node at (8.4,-2.7) {$+$};
				\node at (9.6,-2.7)  {$A_{t,2}\mathbf{f}_0^{(2)}$};
				\node at (10.8,-2.7) {$+$};
				\node at (12.0,-2.7) {$A_{t,3}\mathbf{f}_1^{(2)}$};
				\node at (13.2,-2.7) {$+$};
				\node at (13.7,-2.7) {$\cdots$};
				\node at (14.2,-2.7) {$+$};
				\node at (15.3,-2.7) {$A_{t,11}\mathbf{f}_9^{(2)}$};
				\node at (16.7,-2.7)  {$=\mathbf{0}$};		
				%=======================================Group===========================================%
				\draw[dashed] (-2.5,-0.95) rectangle (2.5,-3.1);
				\draw[dashed] (3.1,-0.95) rectangle (8.1,-3.1);
				\draw[dashed] (8.7,-0.95) rectangle (10.5,-3.1);
				\draw[dashed] (11.1,-0.95) rectangle (12.9,-3.1);
				\draw[dashed] (14.4,-0.95) rectangle (16.2,-3.1);
			\end{tikzpicture}
			\captionsetup{justification=centering}
			\caption{The $t$-th PCG of $(12,8)$ MDS code $\mathcal{C}_3$, where $t\in [0:4)$.}
			\label{fig t-PCG of (12,8) Code C2}
		\end{figure}
		
		Hereafter, we check that the goal nodes of the $(12,8)$ MDS code $\mathcal{C}_3$ have the $3$-optimal access property, where the verification of the MDS property and repair property of remainder nodes can be referred to Theorems \ref{Thr MDS property} and \ref{Thr repair property for RN}, respectively. Let us focus on the repair of the goal node 0 by connecting $d=10$ helper nodes, say remainder nodes 0 to $9$. To this end, the data 
		\begin{eqnarray}\label{eqn the downloaded data for repair GNs}
			(\mathbf{f}_0^{(0)},\mathbf{f}_1^{(0)},\cdots,\mathbf{f}_9^{(0)})
		\end{eqnarray}
		are downloaded, then the PCG-data of remainder nodes 0 to 9 in the first LSE of Figure \ref{fig t-PCG of (12,8) Code C2} are known. Thus, we have
		\begin{eqnarray*}
			A_{t,0}\mathbf{g}_0^{(0)}-A_{t,13}\mathbf{g}_0^{(1)}- A_{t,14}\mathbf{g}_0^{(2)}+A_{t,12}\mathbf{g}_1^{(0)}=*,&t\in [0:4)
		\end{eqnarray*}
		which can solve the data $\mathbf{g}_0^{(0)},\mathbf{g}_0^{(1)}$, and $\mathbf{g}_0^{(2)}$ that are stored at goal node 0, by means of the MDS property of base code $\mathcal{C}_0$. Note that $\gamma(d)=10N=\frac{10}{3}\times 3N$ according to \eqref{eqn the downloaded data for repair GNs}, which achieves the lower bound in \eqref{Eqn_bound_on_gamma} since $d=10$ and the sub-packetization level of the $(12,8)$ code $\mathcal{C}_3$ is $3N$. Moreover, the amount of accessed data is also $10N$ by \eqref{eqn the downloaded data for repair GNs}, which also attains the lower bound in \eqref{Eqn_bound_on_gamma}. Thus we can conclude that goal node 0 of the new $(12,8)$ MDS code $\mathcal{C}_3$ has the $3$-optimal access property.
	\end{Example}
	
	\subsection{MDS Property and Repair Property for Desired Code}
	In this subsection, we check the MDS property, the $\delta$-optimal access property of goal nodes and the repair property of remainder nodes of new $(n,k)$ MDS code $\mathcal{C}_3$.
	
	Recall that  reconstructing  the  original data or repairing  a node, we need to connect $k$ or $d$ nodes. Throughout this subsection, let $\mathcal{D}_0\subset [0:2)$ and $\mathcal{D}_1=\{j_0-2,j_1-2,\cdots,j_{s-1}-2\}\subset [0: n-2)$ respectively be the sets of indices of goal nodes and remainder nodes that are not connected,	 where $|\mathcal{D}_0|+|\mathcal{D}_1|=r$
	for the reconstruction process and $|\mathcal{D}_0|+|\mathcal{D}_1|=n-1-d$ for the repair process, respectively.

	\begin{Theorem}\label{Thr MDS property}
		The new $(n,k)$ MDS code $\mathcal{C}_3$ has the MDS property.
	\end{Theorem}
	\begin{proof}
		It is sufficient to show that the data stored at any $r$ out of the $n$ nodes can be recovered by the data stored at the remaining $k$ nodes.  By the data stored at goal node $j\in [0:2)\setminus\mathcal{D}_0$  and remainder node $l\in [0:n-2)\setminus\mathcal{D}_1$, one can then compute the corresponding PCG-data  in Figure \ref{fig t-PCG of (n,k) Code C3}. 
		
		We discuss the MDS property in three cases. 
		\begin{itemize}
			\item [(i)] When the two goal nodes are connected, i.e., $|\mathcal{D}_0|=0$ and $s=r$, it is obvious that the data stored at the $r$ nodes can be recovered by means of the MDS property of the base code $\mathcal{C}_0$.
			\item [(ii)] When one goal node is not connected, i.e., $|\mathcal{D}_0|=1$ and $s=r-1$, without loss of generality, we assume that the first goal node is not connected, which means that the PCG-data of the second goal node are known. Then we can obtain $(\mathbf{g}_0^{(a)},\mathbf{f}_{j_0-2}^{(a)},\cdots,\mathbf{f}_{j_{r-2}-2}^{(a)})$ for all $a\in [1:\delta)$ from the last $\delta-1$ LSEs of Figure \ref{fig t-PCG of (n,k) Code C3}, due to the MDS property of the base code $\mathcal{C}_0$, from which the data marked with underline and dash underline in the first LSE of Figure \ref{fig t-PCG of (n,k) Code C3} can be cancelled. We can thus get $(\mathbf{g}_0^{(0)},\mathbf{f}_{j_0-2}^{(0)},\cdots,\mathbf{f}_{j_{r-2}-2}^{(0)})$ from the first LSE of Figure \ref{fig t-PCG of (n,k) Code C3}. That is, all the data stored at the $r$ nodes are reconstructed.
			\item [(iii)] When the two goal nodes are not connected, i.e., $|\mathcal{D}_0|=2$ and $s=r-2$, then we can first recover $(\mathbf{g}_0^{(a)},\mathbf{g}_1^{(a)},\mathbf{f}_{j_0-2}^{(a)},\cdots,\mathbf{f}_{j_{r-3}-2}^{(a)})$ for $a\in [2:\delta)$ from the last $\delta-2$ LSEs of Figure \ref{fig t-PCG of (n,k) Code C3} according to the MDS property of base code $\mathcal{C}_0$. We next compute the data that are underlined in Figure \ref{fig t-PCG of (n,k) Code C3} from the recovered data, and then get
			\begin{align}\label{eqn the second LES for the proof of MDS property of code C3}\setlength\arraycolsep{2pt}
				\left\{\begin{array}{ll}
					A_{t,0}\mathbf{g}_0^{(0)}-A_{t,n+1}\mathbf{g}_0^{(1)}+A_{t,n}\mathbf{g}_1^{(0)}+A_{t,j_0}\mathbf{f}_{j_0-2}^{(0)}+\cdots+A_{t,j_{r-3}}\mathbf{f}_{j_{r-3}-2}^{(0)}=*, & t\in [0:r)\\
					A_{t,n+1}\mathbf{g}_0^{(1)}+A_{t,1}\mathbf{g}_1^{(1)}-A_{t,n}\mathbf{g}_1^{(0)}+A_{t,j_0}\mathbf{f}_{j_0-2}^{(1)}+\cdots+A_{t,j_{r-3}}\mathbf{f}_{j_{r-3}-2}^{(1)}=*, & t\in [0:r)
				\end{array}\right.
			\end{align}
			from the first two LSEs of Figure \ref{fig t-PCG of (n,k) Code C3}. By summing the two equations of \eqref{eqn the second LES for the proof of MDS property of code C3}, we have
			\begin{eqnarray*}
				A_{t,0}\mathbf{g}_0^{(0)}+A_{t,1}\mathbf{g}_1^{(1)}+A_{t,j_0}(\mathbf{f}_{j_0-2}^{(0)}+\mathbf{f}_{j_0-2}^{(1)})+\cdots+A_{t,j_{r-3}}(\mathbf{f}_{j_{r-3}-2}^{(0)}+\mathbf{f}_{j_{r-3}-2}^{(1)})=*, & t\in [0:r),
			\end{eqnarray*}
			from which we can get 
			\begin{eqnarray}\label{eqn the result for MDS property}
				(\mathbf{g}_0^{(0)},\mathbf{g}_1^{(1)},\mathbf{f}_{j_0-2}^{(0)}+\mathbf{f}_{j_0-2}^{(1)},\cdots,\mathbf{f}_{j_{r-3}-2}^{(0)}+\mathbf{f}_{j_{r-3}-2}^{(1)}).
			\end{eqnarray}
			Then, after cancelling $\mathbf{g}_0^{(0)}$ in  the first LSE of \eqref{eqn the second LES for the proof of MDS property of code C3},  we are able to solve $$(\mathbf{g}_0^{(1)},\mathbf{g}_1^{(0)},\mathbf{f}_{j_0-2}^{(0)},\cdots,\mathbf{f}_{j_{r-3}-2}^{(0)})$$ from it, which together with \eqref{eqn the result for MDS property} shows that the data stored at all the $r$ nodes  have been recovered.
		\end{itemize}
		
		Collecting the above three cases, we then finish the proof.
	\end{proof}
	
	Next, we examine the repair property of code $\mathcal{C}_3$. 	
	
	\begin{Theorem}\label{Thr repair property for RN}
		For $i\in [2:n)$, if reminder node $i-2$ of the base code $\mathcal{C}_0$ has the $\delta$-optimal access property, then 
		the property can be maintained for remainder node $i-2$  of the new $(n,k)$ MDS code $\mathcal{C}_3$. 
	\end{Theorem}
	\begin{proof}
		Assume remainder node $i-2$ (i.e., node $i\in [2:n)$) of the base code $\mathcal{C}_0$ has the $\delta$-optimal access property, where  its $\delta$-repair matrix and $\delta$-select matrix are  respectively $\frac{N}{\delta}\times N$ full-rank matrices $R_{i,\delta}$ and $S_{i,\delta}$, i.e., the following data
		\begin{eqnarray}\label{eqn the downloaded data of repair RNs}
			\{R_{i,\delta}\mathbf{g}_l^{(a)}|l\in [0:2)\backslash \mathcal{D}_0, a\in [0:\delta)\} \textrm{ and }\{R_{i,\delta}\mathbf{f}_{j}^{(a)}|j \in [0:n-2)\backslash (\mathcal{D}_1\cup\{i-2\}),a\in [0:\delta)\}
		\end{eqnarray}
		are downloaded to repair remainder node $i-2$. Similar to the proof of Theorem \ref{Thr MDS property}, the repair of remainder node $i-2$ of the new $(n,k)$ MDS code $\mathcal{C}_3$ is classified in three cases according to $|\mathcal{D}_0|=0,1$, and $2$. Herein we only investigate the case $|\mathcal{D}_0|=2$, i.e., $\mathcal{D}_0=[0:2)$ and $|\mathcal{D}_1|=r -\delta-2$, since the other cases can be verified in a similar manner. In the case of  $|\mathcal{D}_0|=2$, the repair process of remainder node $i-2$   consists of the following three steps.
		\begin{itemize}
			\item [1)] Since $(\mathbf{g}_0^{(b)},\mathbf{g}_1^{(b)},\mathbf{f}_0^{(b)},\mathbf{f}_1^{(b)},\cdots,\mathbf{f}_{n-2}^{(b)})$ is an instance of code $\mathcal{C}_1$,  by Lemma \ref{lem the requirement of obtaining fi} and \eqref{Eqn linear equations eq 1}, we can recover $\mathbf{f}_{i-2}^{(b)}$, $R_{i,\delta}\mathbf{g}_0^{(b)}$ and $R_{i,\delta}\mathbf{g}_1^{(b)}$ for all $b\in [2:\delta)$.
			\item [2)] Multiply $S_{i,\delta}$ on both sides of each of the first two LSEs in Figure \ref{fig t-PCG of (n,k) Code C3},% to solve the data $\mathbf{f}_{i-2}^{(0)}$ and $\mathbf{f}_{i-2}^{(1)}$
			then we obtain the following equations,
			\begin{eqnarray}\label{eqn the first LSE used to solve remainder node 1}
				\left\{\begin{array}{ll}
					S_{i,\delta}A_{t,i}\mathbf{f}_{i-2}^{(0)}+S_{i,\delta}(A_{t,0}\mathbf{g}_0^{(0)}-\sum\limits_{u=1}^{\delta-1}A_{t,n+u}\mathbf{g}_0^{(u)})+S_{i,\delta}A_{t,n}\mathbf{g}_1^{(0)}\\
					\hspace{7.5cm}+\sum\limits_{j=2,j\ne i}^{n-1}S_{i,\delta}A_{t,j}\mathbf{f}_{j-2}^{(0)}=\mathbf{0}, & t\in [0:r)\\
					S_{i,\delta}A_{t,i}\mathbf{f}_{i-2}^{(1)}+S_{i,\delta}A_{t,n+1}\mathbf{g}_0^{(1)}+S_{i,\delta}(A_{t,1}\mathbf{g}_1^{(1)}-\sum\limits_{u=0,u \ne 1}^{\delta-1}A_{t,n+u}\mathbf{g}_0^{(u)})\\
					\hspace{7.5cm}+\sum\limits_{j=2,j\ne i}^{n-1}S_{i,\delta}A_{t,j}\mathbf{f}_{j-2}^{(1)}=\mathbf{0}, & t\in [0:r)
				\end{array}\right.
			\end{eqnarray}
			Since node $i$ of  base code $\mathcal{C}_0$ has the $\delta$-optimal access property, first by  applying \eqref{repair node requirement3} and then  substituting  the downloaded/recovered data we can then simplify \eqref{eqn the first LSE used to solve remainder node 1} as
			\begin{align}\label{eqn the first LSE used to solve remainder node 3}
				\left\{\begin{array}{ll}
					S_{i,\delta}A_{t,i}\mathbf{f}_{i-2}^{(0)}+\tilde{A}_{t,0,i,\delta}R_{i,\delta}\mathbf{g}_0^{(0)}-\tilde{A}_{t,n+1,i,\delta}R_{i,\delta}\mathbf{g}_0^{(1)}+\tilde{A}_{t,n,i,\delta}R_{i,\delta}\mathbf{g}_1^{(0)}\\
					\hspace{6.5cm}+\sum\limits_{z=0}^{r-\delta-3}\tilde{A}_{t,j_z,i,\delta}R_{i,\delta}\mathbf{f}_{j_z-2}^{(0)}=*, & t\in [0:r)\\
					S_{i,\delta}A_{t,i}\mathbf{f}_{i-2}^{(1)}+\tilde{A}_{t,n+1,i,\delta}R_{i,\delta}\mathbf{g}_0^{(1)}+\tilde{A}_{t,1,i,\delta}R_{i,\delta}\mathbf{g}_1^{(1)}-\tilde{A}_{t,n,i,\delta}R_{i,\delta}\mathbf{g}_1^{(0)}\\
					\hspace{6.5cm}+\sum\limits_{z=0}^{r-\delta-3}\tilde{A}_{t,j_z,i,\delta}R_{i,\delta}\mathbf{f}_{j_z-2}^{(1)}=*, & t\in [0:r)
				\end{array}\right.
			\end{align}
			
			\item [3)] Recover the data $\mathbf{f}_{i-2}^{(0)}$ and $\mathbf{f}_{i-2}^{(1)}$ from \eqref{eqn the first LSE used to solve remainder node 3}. In fact, this step is almost the same as Case (iii) in the proof of Theorem \ref{Thr MDS property}, i.e., we recover the desired data by summing the two equations in  \eqref{eqn the first LSE used to solve remainder node 3} and then applying   Lemma \ref{lem the requirement of obtaining fi}.
		\end{itemize}	
		
		It follows from \eqref{eqn the downloaded data of repair RNs} that both the repair bandwidth and the amount of accessed data are $d\delta\cdot\textrm{Rank}(R_{i,\delta})=\frac{d \delta N}{\delta}=\frac{d \delta N}{d-k+1}$ due to $d=k+\delta-1$, which attains the lower bound in  \eqref{Eqn_bound_on_gamma} by the fact that the sub-packetization level of code $\mathcal{C}_3$ is $\delta N$. This finishes the proof.	\end{proof}
	
	\begin{Theorem}\label{Thr repair property for GN}
		The new $(n,k)$ MDS code $\mathcal{C}_3$ has the $\delta$-optimal access property for the two goal nodes.
	\end{Theorem}
	\begin{proof}
		Let us consider the repair of goal node $i\in [0:2)$ by connecting $d=k+\delta-1$ surviving nodes, where the $(i+1)$-th LSE of Figure \ref{fig t-PCG of (n,k) Code C3} is used to recover the data stored at goal node $i$. By downloading $\mathbf{g}_l^{(i)}$ and $\mathbf{f}_j^{(i)}$ from goal node $l\in [0:2)\backslash (\mathcal{D}_0\cup \{i\})$ and remainder node $j\in [0:n-2)\backslash \mathcal{D}_1$, then the $(i+1)$-th LSE of Figure \ref{fig t-PCG of (n,k) Code C3} can be reduced to
		\begin{eqnarray}\label{eqn the second LSE for repair GN}
			A_{t,i}\mathbf{g}_i^{(i)}-\sum\limits_{u=0,u\ne i}^{\delta-1}A_{t,n+u}\mathbf{g}_i^{(u)}+\sum\limits_{l\in \mathcal{D}_0}A_{t,n+i}\mathbf{g}_l^{(i)}+\sum\limits_{z=0}^{ s-1}A_{t,j_z}\mathbf{f}_{j_z-2}^{(i)}=*, &t\in [0:r),
		\end{eqnarray}
		where $s=|\mathcal{D}_1|$. Obviously, there are
		$$\delta N+(|\mathcal{D}_0|+|\mathcal{D}_1|)N=\delta N+(n-d-1)N=\delta N+(r-\delta)N=rN$$ unknown variables in \eqref{eqn the second LSE for repair GN}, which implies that the data $\mathbf{g}_i^{(a)}$ for $a\in [0:\delta)$ can be solved by means of the MDS property of base code $\mathcal{C}_0$.  
		
		It is clear that the amount of accessed data and the repair bandwidth are the same, i.e., $dN=\frac{d\delta N}{d-k+1}$ because of $d=k+\delta-1$, which is $\delta$-optimal access according to the lower bound in  \eqref{Eqn_bound_on_gamma}.
	\end{proof}

	\section{New MDS Codes with $\delta$-Optimal Access Property for All Nodes}\label{section algorithm for all nodes dz-optimal access}
	In this section, we recursively  apply  the proposed generic transformation  to  an $(n',k'=n'-r)$  MDS code  to generate an $(n, k=n-r)$ MDS code with the $\delta$-optimal access property for all nodes. Specifically, by choosing an $(n'=n+\delta\lceil\frac{n}{2} \rceil,k'=k+\delta\lceil\frac{n}{2}\rceil)$ scalar MDS code over $\mathbb{F}_q$ as the base code, such as the well-known Reed-Solomon code, the $(n,k)$ desired code can be generated through  Algorithm \ref{Alg_repair_for_all_nodes}, i.e., by recursively  using the  generic transformation $\lceil\frac{n}{2}\rceil$ times.
	
	According to Theorems \ref{Thr MDS property}-\ref{Thr repair property for GN}, the following result can be obtained directly.
	\begin{Theorem}\label{Cro the desired code C}
		By choosing any $(n'=n+\delta\lceil\frac{n}{2} \rceil,k'=k+\delta\lceil\frac{n}{2}\rceil)$ scalar MDS code as base code, applying Algorithm \ref{Alg_repair_for_all_nodes}, one can get a new $(n,k)$ MDS code $\mathcal{C}$ with $\delta$-optimal access property for all nodes over $\mathbb{F}_q$ with $q\ge n+\delta\lceil\frac{n}{2} \rceil$, where the sub-packetization level of the new $(n,k)$ MDS code is $\delta^{\lceil\frac{n}{2} \rceil}$.
	\end{Theorem}
	
	It is worthy noting that the sub-packetization level of the new $(n,k)$ MDS code $\mathcal{C}$ is much smaller than that of the two codes proposed in \cite{Ye_barg_1},  as shown in Table \ref{Table comp Ye Kumar}.
	
	\begin{algorithm}[!htbp]
		\caption{The method to generate an MDS code with the $\delta$-optimal access property for all nodes}\label{Alg_repair_for_all_nodes}
		\begin{algorithmic}[1]
			\Require The followings are required to input:
			\begin{enumerate}
				\item  The parameters $n$, $k$ and $\delta$;
				\item A known $(n',k')$ scalar MDS code $\mathcal{Q}_0$ over $\mathbb{F}_q$, where $k'=k+\tau\delta$ and $n'=n+\tau\delta$ with $\tau=\lceil\frac{n}{2} \rceil$.
			\end{enumerate}
			\Ensure The desired $(n,k)$ MDS code $\mathcal{Q}_\tau$ of sub-packetization level $\delta^\tau $ with $\delta$-optimal access property for all nodes.
			\For{$t=0$; $t < \tau$; $t++$}
			\State Set the $(n_t,k_t)$ MDS code $\mathcal{Q}_t$ as the base code where $n_t=n'-t\delta$ and $k_t=k'-t\delta$
			\If {$t< \tau-1$}
			\State Designate the two nodes $2t,2t+1$ as the goal nodes
			\Else
			\State Designate the two nodes $n-2,n-1$ as the goal nodes
			\EndIf
			\State Applying the generic transformation to code $\mathcal{Q}_t$ to get a new code $\mathcal{Q}_{t+1}$ with sub-packetization level $\delta^{t+1}$.
			\EndFor
		\end{algorithmic}
	\end{algorithm}

	In the following, we provide an example of Algorithm \ref{Alg_repair_for_all_nodes}.
	\begin{Example}
		Our ultimate goal is to obtain an $(8,5)$ MDS code with $2$-optimal access for all nodes through four rounds of the generic transformations according to Algorithm \ref{Alg_repair_for_all_nodes}, where a $(16,13)$ scalar MDS code over $\mathbb{F}_{17}$  is chosen as the base code $\mathcal{Q}_0$. The $t$-th PCG of the base code $\mathcal{Q}_0$ is given follows:
		\begin{eqnarray*}
			f_{0,0}+2^tf_{1,0}+\cdots+16^t f_{15,0}=0, & t\in [0:3).
		\end{eqnarray*}
		
		Through the four rounds of generic transformations, we obtain the codes $\mathcal{Q}_1$, $\mathcal{Q}_2$, $\mathcal{Q}_3$ and $\mathcal{Q}_4$, where the $t$-th PCG of codes $\mathcal{Q}_1$, $\mathcal{Q}_2$, $\mathcal{Q}_3$ and $\mathcal{Q}_4$  are respectively given by Figures \ref{fig the PCG of Q1}, \ref{fig the PCG of Q2}, \ref{fig the PCG of Q3} and \ref{fig the PCG of Q4}. For convenience, the $a$-th symbol of node $i$ in each one of the codes $\mathcal{Q}_0,\mathcal{Q}_1,\cdots,\mathcal{Q}_4$ is always represented by $f_{i,a}$. Moreover, in the first, second, third and fourth rounds, nodes 0, 1, nodes 2, 3, nodes 4, 5 and nodes 6, 7 are chosen as the goal nodes, respectively.
		%=================================PCG of Q1===================================================%
		\begin{figure}[!htbp]
			\tikzset{
				source1/.style={draw,thick,rounded corners,align=center,inner sep=.1cm,minimum width=1.3cm,minimum height=1.0cm},
				source2/.style={draw,thick,rounded corners,align=center,inner sep=.1cm,minimum width=2.7cm,minimum height=1.0cm},
				global scale/.style={scale=#1,every node/.append style={scale=#1}}
			}
			\centering
			\begin{tikzpicture}[global scale=0.7]
				%=======================================node number=======================================================%
				\node [source2] (GN 0) at (-0.1,0.1)  {PCG-data \\ of node 0};
				\node [source2] (GN 2) at (3.2,0.1)  {PCG-data \\ of node 1};
				\node [source1] (RN 0) at (6.1,0.1)  {PCG-data \\ of node 2};
				\node [source1] (RN 0) at (8.6,0.1)  {PCG-data \\ of node 3};
				\node  at (10.2,0.1)  {$\cdots$};
				\node [source1] (RN 0) at (12.1,0.1)  {PCG-data \\ of node 12};
				\node [source1] (RN 0) at (14.8,0.1)  {PCG-data \\ of node 13};
				%====================================PCG-data of first row===================================================%
				\node  at (-0.1,-1.2)  {$f_{0,0}-16^t f_{0,1}$};
				\node  at (1.6,-1.2) {$+$};
				\node  at (3.2,-1.2)  {$15^tf_{1,0}$};
				\node  at (4.9,-1.2) {$+$};
				\node  at (6.1,-1.2)  {$3^tf_{2,0}$};
				\node  at (7.35,-1.2) {$+$};
				\node  at (8.6,-1.2)  {$4^tf_{3,0}$};
				\node  at (9.75,-1.2) {$+$};
				\node  at (10.3,-1.2)  {$\cdots$};
				\node  at (10.8,-1.2) {$+$};
				\node  at (12.1,-1.2)  {$13^tf_{12,0}$};
				\node  at (13.4,-1.2) {$+$};
				\node  at (14.8,-1.2)  {$14^tf_{13,0}$};
				\node  at (16.3,-1.2)  {$=0$};
				%====================================PCG-data of second row===================================================%
				\node  at (-0.1,-1.8)  {$16^t f_{0,1}$};
				\node  at (1.6,-1.8) {$+$};
				\node  at (3.2,-1.8)  {$2^tf_{1,1}-15^tf_{1,0}$};
				\node  at (4.9,-1.8) {$+$};
				\node  at (6.1,-1.8)  {$3^tf_{2,1}$};
				\node  at (7.35,-1.8) {$+$};
				\node  at (8.6,-1.8)  {$4^tf_{3,1}$};
				\node  at (9.75,-1.8) {$+$};
				\node  at (10.3,-1.8) {$\cdots$};
				\node  at (10.8,-1.8) {$+$};
				\node  at (12.1,-1.8)  {$13^tf_{12,1}$};
				\node  at (13.4,-1.8) {$+$};
				\node  at (14.8,-1.8)  {$14^tf_{13,1}$};
				\node  at (16.3,-1.8)  {$=0$};
				%=======================================Group===========================================%
				\draw[dashed] (-1.4,-0.8) rectangle (1.2,-2.2);
				\draw[dashed] (1.85,-0.8) rectangle (4.55,-2.2);
				\draw[dashed] (5.2,-0.8) rectangle (7.0,-2.2);
				\draw[dashed] (7.7,-0.8) rectangle (9.5,-2.2);
				\draw[dashed] (11.1,-0.8) rectangle (13.1,-2.2);
				\draw[dashed] (13.8,-0.8) rectangle (15.75,-2.2);
			\end{tikzpicture}
			\captionsetup{justification=centering}
			\caption{The $t$-th PCG of the $(14,11)$ MDS code $\mathcal{Q}_1$, where $t\in [0:3)$ and nodes 0, 1 are chosen as goal nodes.}
			\label{fig the PCG of Q1}
		\end{figure}
		%====================================PCG of Q2=============================================%
		\begin{figure}[!htbp]
			\tikzset{
				source1/.style={draw,thick,rounded corners,align=center,inner sep=.1cm,minimum width=1.3cm,minimum height=0.8cm},
				source2/.style={draw,thick,rounded corners,align=center,inner sep=.1cm,minimum width=2.4cm,minimum height=0.8cm},
				global scale/.style={scale=#1,every node/.append style={scale=#1}}
			}
			\centering
			\begin{tikzpicture}[global scale=0.7]
				%=======================================node number=======================================================%
				\node [source2] (GN 0) at (0,0.2)  { PCG-data\\ of node 0};
				\node [source2] (GN 2) at (2.9,0.2)  {PCG-data\\ of node 1};
				\node [source2] (RN 0) at (5.8,0.2)  {PCG-data\\ of node 2};
				\node [source2] (RN 0) at (8.7,0.2)  { PCG-data\\ of node 3};
				\node [source1] (RN 0) at (11.4,0.2)  {PCG-data\\ of node 4};
				\node [source1] (RN 0) at (13.8,0.2)  {PCG-data\\ of node 5};
				%\node at (14.9,0.1){$+$};
				\node  at (15.5,0.2)  { $\cdots$};
				%\node at (15.8,0.1) {$+$};
				\node [source1] (RN 0) at (17.2,0.2)  {PCG-data\\ of node 10};
				\node [source1] (RN 0) at (19.8,0.2)  {PCG-data\\ of node 11};
				%\node  at (20.0,0.1)  { $=0$};
				%====================================PCG-data of first row===================================================%
				\node  at (0,-1)  {\small $f_{0,0}-16^t f_{0,1}$};
				\node  at (1.45,-1) {\small $+$};
				\node  at (2.9,-1)  {\small $15^tf_{1,0}$};
				\node  at (4.35,-1) {\small $+$};
				\node  at (5.8,-1)  {\small $3^tf_{2,0}-14^tf_{2,2}$};
				\node  at (7.25,-1) {\small $+$};
				\node  at (8.7,-1)  {\small $13^tf_{3,0}$};
				\node  at (10.15,-1) {\small $+$};
				\node  at (11.4,-1)  {\small $5^tf_{4,0}$};
				\node  at (12.6,-1) {\small $+$};
				\node  at (13.8,-1)  {\small $6^tf_{5,0}$};
				\node  at (15.0,-1) {\small $+$};
				\node  at (15.5,-1)  {\small $\cdots$};
				\node  at (15.9,-1)  {\small $+$};
				\node  at (17.2,-1)  {\small $11^tf_{10,0}$};
				\node  at (18.5,-1)  {\small $+$};
				\node  at (19.8,-1)  {\small $12^tf_{11,0}$};
				\node  at (21.2,-1)  {\small $=0$};
				%====================================PCG-data of second row===================================================%
				\node  at (0,-1.6)  {\small $16^tf_{0,1}$};
				\node  at (1.45,-1.6) {\small $+$};
				\node  at (2.9,-1.6)  {\small $2^tf_{1,1}-15^tf_{1,0}$};
				\node  at (4.35,-1.6) {\small $+$};
				\node  at (5.8,-1.6)  {\small $3^tf_{2,1}-14^tf_{2,3}$};
				\node  at (7.25,-1.6) {\small $+$};
				\node  at (8.7,-1.6)  {\small $13^tf_{3,1}$};
				\node  at (10.15,-1.6) {\small $+$};
				\node  at (11.4,-1.6) {\small $5^tf_{4,1}$};
				\node  at (12.6,-1.6) {\small $+$};
				\node  at (13.8,-1.6)  {\small $6^tf_{5,1}$};
				\node  at (15.0,-1.6) {\small $+$};
				\node  at (15.5,-1.6)  {\small $\cdots$};
				\node  at (15.9,-1.6)  {\small $+$};
				\node  at (17.2,-1.6)  {\small $11^tf_{10,0}$};
				\node  at (18.5,-1.6)  {\small $+$};
				\node  at (19.8,-1.6)  {\small $12^tf_{11,0}$};
				\node  at (21.2,-1.6)  {\small $=0$};
				%====================================PCG-data of first row===================================================%
				\node  at (0,-2.2)  {\small $ f_{0,2}-16^t f_{0,3}$};
				\node  at (1.45,-2.2) {\small $+$};
				\node  at (2.9,-2.2)  {\small $15^tf_{1,2}$};
				\node  at (4.35,-2.2) {\small $+$};
				\node  at (5.8,-2.2)  {\small $14^tf_{2,2}$};
				\node  at (7.25,-2.2) {\small $+$};
				\node  at (8.7,-2.2)  {\small $4^tf_{3,2}-13^tf_{3,0}$};
				\node  at (10.15,-2.2) {\small $+$};
				\node  at (11.4,-2.2)  {\small $5^tf_{4,2}$};
				\node  at (12.6,-2.2) {\small $+$};
				\node  at (13.8,-2.2)  {\small $6^tf_{5,2}$};
				\node  at (15.0,-2.2) {\small $+$};
				\node  at (15.5,-2.2)  {\small $\cdots$};
				\node  at (15.9,-2.2)  {\small $+$};
				\node  at (17.2,-2.2)  {\small $11^tf_{10,0}$};
				\node  at (18.5,-2.2)  {\small $+$};
				\node  at (19.8,-2.2)  {\small $12^tf_{11,0}$};
				\node  at (21.2,-2.2)  {\small $=0$};
				%====================================PCG-data of second row===================================================%
				\node  at (0,-2.8)  {\small $16^t f_{0,3}$};
				\node  at (1.45,-2.8) {\small $+$};
				\node  at (2.9,-2.8)  {\small $2^tf_{1,3}-15^tf_{1,2}$};
				\node  at (4.35,-2.8) {\small $+$};
				\node  at (5.8,-2.8)  {\small $14^tf_{2,3}$};
				\node  at (7.25,-2.8) {\small $+$};
				\node  at (8.7,-2.8)  {\small $4^tf_{3,3}-13^tf_{3,1}$};
				\node  at (10.15,-2.8) {\small $+$};
				\node  at (11.4,-2.8) {\small $5^tf_{4,3}$};
				\node  at (12.6,-2.8) {\small $+$};
				\node  at (13.8,-2.8)  {\small $6^tf_{5,3}$};
				\node  at (15.0,-2.8) {\small $+$};
				\node  at (15.5,-2.8)  {\small $\cdots$};
				\node  at (15.9,-2.8)  {\small $+$};
				\node  at (17.2,-2.8)  {\small $11^tf_{10,0}$};
				\node  at (18.5,-2.8)  {\small $+$};
				\node at (19.8,-2.8)  {\small $12^tf_{11,0}$};
				\node  at (21.2,-2.8)  {\small $=0$};
				%=======================================Group===========================================%
				\draw[dashed] (-1.2,-0.7) rectangle (1.15,-3.1);
				\draw[dashed] (1.7,-0.7) rectangle (4.1,-3.1);
				\draw[dashed] (4.6,-0.7) rectangle (7.0,-3.1);
				\draw[dashed] (7.5,-0.7) rectangle (9.9,-3.1);
				\draw[dashed] (10.5,-0.7) rectangle (12.3,-3.1);
				\draw[dashed] (12.9,-0.7) rectangle (14.7,-3.1);
				\draw[dashed] (16.2,-0.7) rectangle (18.2,-3.1);
				\draw[dashed] (18.8,-0.7) rectangle (20.8,-3.1);
			\end{tikzpicture}
			\captionsetup{justification=centering}
			\caption{The $t$-th PCG of the $(12,9)$ MDS code $\mathcal{Q}_2$, where $t\in [0:3)$ and nodes 2, 3 are chosen as goal nodes}
			\label{fig the PCG of Q2}
		\end{figure}
		%====================================PCG of Q3=============================================%
		\begin{figure}[!htbp]
			\tikzset{
				source1/.style={draw,thick,rounded corners,align=center,inner sep=.1cm,minimum width=1.3cm,minimum height=0.8cm},
				source2/.style={draw,thick,rounded corners,align=center,inner sep=.1cm,minimum width=2.4cm,minimum height=0.8cm},
				global scale/.style={scale=#1,every node/.append style={scale=#1}}
			}
			\centering
			\begin{tikzpicture}[global scale=0.61]
				%=======================================node number=======================================================%
				\node [source2] (GN 0) at (0,0.2)  { PCG-data\\ of node 0};
				%\node at (1.3,0.1){$+$};
				\node [source2] (GN 2) at (2.9,0.2)  {PCG-data\\ of node 1};
				\node [source2] (RN 0) at (5.8,0.2)  {PCG-data\\ of node 2};
				\node [source2] (RN 0) at (8.7,0.2)  {PCG-data\\ of node 3};
				\node [source2] (RN 0) at (11.6,0.2)  {PCG-data\\ of node 4};
				\node [source2] (RN 0) at (14.5,0.2)  {PCG-data\\ of node 5};
				\node [source1] (RN 0) at (17.1,0.2)  {PCG-data\\ of node 6};
				\node [source1] (RN 0) at (19.5,0.2)  {PCG-data\\ of node 7};
				\node [source1] (RN 0) at (21.9,0.2)  {PCG-data\\ of node 8};
				\node [source1] (RN 0) at (24.3,0.2)  {PCG-data\\ of node 9};
				%====================================PCG-data of first row===================================================%
				\node  at (0,-1)  {\small $f_{0,0}-16^t f_{0,1}$};
				\node  at (1.45,-1) {\small $+$};
				\node  at (2.9,-1)  {\small $15^tf_{1,0}$};
				\node  at (4.35,-1) {\small $+$};
				\node  at (5.8,-1)  {\small $3^tf_{2,0}-14^tf_{2,2}$};
				\node  at (7.25,-1) {\small $+$};
				\node  at (8.7,-1)  {\small $13^tf_{3,0}$};
				\node  at (10.15,-1) {\small $+$};
				\node  at (11.6,-1)  {\small $5^tf_{4,0}-12^tf_{4,4}$};
				\node  at (13.05,-1) {\small $+$};
				\node  at (14.5,-1)  {\small $11^tf_{5,0}$};
				\node  at (15.9,-1) {\small $+$};
				\node  at (17.1,-1)  {\small $7^tf_{6,0}$};
				\node  at (18.3,-1)  {\small $+$};
				\node  at (19.5,-1)  {\small $8^tf_{7,0}$};
				\node  at (20.7,-1)  {\small $+$};
				\node  at (21.9,-1)  {\small $9^tf_{8,0}$};
				\node  at (23.1,-1)  {\small $+$};
				\node  at (24.3,-1)  {\small $10^tf_{9,0}$};
				\node  at (25.7,-1)  {\small $=0$};
				%====================================PCG-data of second row===================================================%
				\node  at (0,-1.6)  {\small $16^tf_{0,1}$};
				\node  at (1.45,-1.6) {\small $+$};
				\node  at (2.9,-1.6)  {\small $2^tf_{1,1}-15^tf_{1,0}$};
				\node  at (4.35,-1.6) {\small $+$};
				\node  at (5.8,-1.6)  {\small $3^tf_{2,1}-14^tf_{2,3}$};
				\node  at (7.25,-1.6) {\small $+$};
				\node  at (8.7,-1.6)  {\small $13^tf_{3,1}$};
				\node  at (10.15,-1.6) {\small $+$};
				\node  at (11.6,-1.6) {\small $5^tf_{4,1}-12^tf_{4,5}$};
				\node  at (13.05,-1.6) {\small $+$};
				\node  at (14.5,-1.6)  {\small $11^tf_{5,1}$};
				\node  at (15.9,-1.6) {\small $+$};
				\node  at (17.1,-1.6)  {\small $7^tf_{6,1}$};
				\node  at (18.3,-1.6)  {\small $+$};
				\node  at (19.5,-1.6)  {\small $8^tf_{7,1}$};
				\node  at (20.7,-1.6)  {\small $+$};
				\node  at (21.9,-1.6)  {\small $9^tf_{8,1}$};
				\node  at (23.1,-1.6)  {\small $+$};
				\node  at (24.3,-1.6)  {\small $10^tf_{9,1}$};
				\node  at (25.7,-1.6)  {\small $=0$};
				%====================================PCG-data of third row===================================================%
				\node  at (0,-2.2)  {\small $ f_{0,2}\hspace{-0.5mm}-\hspace{-0.5mm}16^t f_{0,3}$};
				\node  at (1.45,-2.2) {\small $+$};
				\node  at (2.9,-2.2)  {\small $15^tf_{1,2}$};
				\node  at (4.35,-2.2) {\small $+$};
				\node  at (5.8,-2.2)  {\small $14^tf_{2,2}$};
				\node  at (7.25,-2.2) {\small $+$};
				\node  at (8.7,-2.2)  {\small $4^tf_{3,2}\hspace{-0.5mm}-\hspace{-0.5mm}13^tf_{3,0}$};
				\node  at (10.15,-2.2) {\small $+$};
				\node  at (11.6,-2.2)  {\small $5^tf_{4,2}\hspace{-0.5mm}-\hspace{-0.5mm}12^tf_{4,6}$};
				\node  at (13.05,-2.2) {\small $+$};
				\node  at (14.5,-2.2)  {\small $11^tf_{5,2}$};
				\node  at (15.9,-2.2) {\small $+$};
				\node  at (17.1,-2.2)  {\small $7^tf_{6,2}$};
				\node  at (18.3,-2.2)  {\small $+$};
				\node  at (19.5,-2.2)  {\small $8^tf_{7,2}$};
				\node  at (20.7,-2.2)  {\small $+$};
				\node  at (21.9,-2.2)  {\small $9^tf_{8,2}$};
				\node  at (23.1,-2.2)  {\small $+$};
				\node  at (24.3,-2.2)  {\small $10^tf_{9,2}$};
				\node  at (25.7,-2.2)  {\small $=0$};
				%====================================PCG-data of fourth row===================================================%
				\node  at (0,-2.8)  {\small $16^t f_{0,3}$};
				\node  at (1.45,-2.8) {\small $+$};
				\node  at (2.9,-2.8)  {\small $2^tf_{1,3}\hspace{-0.5mm}-\hspace{-0.5mm}15^tf_{1,2}$};
				\node  at (4.35,-2.8) {\small $+$};
				\node  at (5.8,-2.8)  {\small $14^tf_{2,3}$};
				\node  at (7.25,-2.8) {\small $+$};
				\node  at (8.7,-2.8)  {\small $4^tf_{3,3}\hspace{-0.5mm}-\hspace{-0.5mm}13^tf_{3,1}$};
				\node  at (10.15,-2.8) {\small $+$};
				\node  at (11.6,-2.8) {\small $5^tf_{4,3}\hspace{-0.5mm}-\hspace{-0.5mm}12^tf_{4,7}$};
				\node  at (13.05,-2.8) {\small $+$};
				\node  at (14.5,-2.8)  {\small $11^tf_{5,3}$};
				\node  at (15.9,-2.8) {\small $+$};
				\node  at (17.1,-2.8)  {\small $7^tf_{6,3}$};
				\node  at (18.3,-2.8)  {\small $+$};
				\node  at (19.5,-2.8)  {\small $8^tf_{7,3}$};
				\node  at (20.7,-2.8)  {\small $+$};
				\node  at (21.9,-2.8)  {\small $9^tf_{8,3}$};
				\node  at (23.1,-2.8)  {\small $+$};
				\node  at (24.3,-2.8)  {\small $10^tf_{9,3}$};
				\node  at (25.7,-2.8)  {\small $=0$};
				%====================================PCG-data of fifth row===================================================%
				\node  at (0,-3.4)  {\small $f_{0,4}-16^t f_{0,5}$};
				\node  at (1.45,-3.4) {\small $+$};
				\node  at (2.9,-3.4)  {\small $15^tf_{1,4}$};
				\node  at (4.35,-3.4) {\small $+$};
				\node  at (5.8,-3.4)  {\small $3^tf_{2,4}\hspace{-0.5mm}-\hspace{-0.5mm}14^tf_{2,6}$};
				\node  at (7.25,-3.4) {\small $+$};
				\node  at (8.7,-3.4)  {\small $13^tf_{3,4}$};
				\node  at (10.15,-3.4) {\small $+$};
				\node  at (11.6,-3.4)  {\small $12^tf_{4,4}$};
				\node  at (13.05,-3.4) {\small $+$};
				\node  at (14.5,-3.4)  {\small $6^tf_{5,4}\hspace{-0.5mm}-\hspace{-0.5mm}11^tf_{5,0}$};
				\node  at (15.9,-3.4) {\small $+$};
				\node  at (17.1,-3.4)  {\small $7^tf_{6,4}$};
				\node  at (18.3,-3.4)  {\small $+$};
				\node  at (19.5,-3.4)  {\small $8^tf_{7,4}$};
				\node  at (20.7,-3.4)  {\small $+$};
				\node  at (21.9,-3.4)  {\small $9^tf_{8,4}$};
				\node  at (23.1,-3.4)  {\small $+$};
				\node  at (24.3,-3.4)  {\small $10^tf_{9,4}$};
				\node  at (25.7,-3.4)  {\small $=0$};
				%====================================PCG-data of sixth row===================================================%
				\node  at (0,-4.0)  {\small $16^tf_{0,5}$};
				\node  at (1.45,-4.0) {\small $+$};
				\node  at (2.9,-4.0)  {\small $2^tf_{1,5}\hspace{-0.5mm}-\hspace{-0.5mm}15^tf_{1,4}$};
				\node  at (4.35,-4.0) {\small $+$};
				\node  at (5.8,-4.0)  {\small $3^tf_{2,5}\hspace{-0.5mm}-\hspace{-0.5mm}14^tf_{2,7}$};
				\node  at (7.25,-4.0) {\small $+$};
				\node  at (8.7,-4.0)  {\small $13^tf_{3,5}$};
				\node  at (10.15,-4.0) {\small $+$};
				\node  at (11.6,-4.0) {\small $12^tf_{4,5}$};
				\node  at (13.05,-4.0) {\small $+$};
				\node  at (14.5,-4.0)  {\small $6^tf_{5,5}\hspace{-0.5mm}-\hspace{-0.5mm}11^tf_{5,1}$};
				\node  at (15.9,-4.0) {\small $+$};
				\node  at (17.1,-4.0)  {\small $7^tf_{6,5}$};
				\node  at (18.3,-4.0)  {\small $+$};
				\node  at (19.5,-4.0)  {\small $8^tf_{7,5}$};
				\node  at (20.7,-4.0)  {\small $+$};
				\node  at (21.9,-4.0)  {\small $9^tf_{8,5}$};
				\node  at (23.1,-4.0)  {\small $+$};
				\node  at (24.3,-4.0)  {\small $10^tf_{9,5}$};
				\node  at (25.7,-4.0)  {\small $=0$};
				%====================================PCG-data of seventh row===================================================%
				\node  at (0,-4.6)  {\small $ f_{0,6}-16^t f_{0,7}$};
				\node  at (1.45,-4.6) {\small $+$};
				\node  at (2.9,-4.6)  {\small $15^tf_{1,6}$};
				\node  at (4.35,-4.6) {\small $+$};
				\node  at (5.8,-4.6)  {\small $14^tf_{2,6}$};
				\node  at (7.25,-4.6) {\small $+$};
				\node  at (8.7,-4.6)  {\small $4^tf_{3,6}\hspace{-0.5mm}-\hspace{-0.5mm}13^tf_{3,4}$};
				\node  at (10.15,-4.6) {\small $+$};
				\node  at (11.6,-4.6)  {\small $12^tf_{4,6}$};
				\node  at (13.05,-4.6) {\small $+$};
				\node  at (14.5,-4.6)  {\small $6^tf_{5,6}\hspace{-0.5mm}-\hspace{-0.5mm}11^tf_{5,2}$};
				\node  at (15.9,-4.6) {\small $+$};
				\node  at (17.1,-4.6)  {\small $7^tf_{6,6}$};
				\node  at (18.3,-4.6)  {\small $+$};
				\node  at (19.5,-4.6)  {\small $8^tf_{7,6}$};
				\node  at (20.7,-4.6)  {\small $+$};
				\node  at (21.9,-4.6)  {\small $9^tf_{8,6}$};
				\node  at (23.1,-4.6)  {\small $+$};
				\node  at (24.3,-4.6)  {\small $10^tf_{9,6}$};
				\node  at (25.7,-4.6)  {\small $=0$};
				%====================================PCG-data of eighth row===================================================%
				\node  at (0,-5.2)  {\small $16^t f_{0,7}$};
				\node  at (1.45,-5.2) {\small $+$};
				\node  at (2.9,-5.2)  {\small $2^tf_{1,7}\hspace{-0.5mm}-\hspace{-0.5mm}15^tf_{1,6}$};
				\node  at (4.35,-5.2) {\small $+$};
				\node  at (5.8,-5.2)  {\small $14^tf_{2,7}$};
				\node  at (7.25,-5.2) {\small $+$};
				\node  at (8.7,-5.2)  {\small $4^tf_{3,7}\hspace{-0.5mm}-\hspace{-0.5mm}13^tf_{3,5}$};
				\node  at (10.15,-5.2) {\small $+$};
				\node  at (11.6,-5.2) {\small $12^tf_{4,7}$};
				\node  at (13.05,-5.2) {\small $+$};
				\node  at (14.5,-5.2)  {\small $6^tf_{5,7}\hspace{-0.5mm}-\hspace{-0.5mm}11^tf_{5,3}$};
				\node  at (15.9,-5.2) {\small $+$};
				\node  at (17.1,-5.2)  {\small $7^tf_{6,7}$};
				\node  at (18.3,-5.2)  {\small $+$};
				\node  at (19.5,-5.2)  {\small $8^tf_{7,7}$};
				\node  at (20.7,-5.2)  {\small $+$};
				\node  at (21.9,-5.2)  {\small $9^tf_{8,7}$};
				\node  at (23.1,-5.2)  {\small $+$};
				\node  at (24.3,-5.2)  {\small $10^tf_{9,7}$};
				\node  at (25.7,-5.2)  {\small $=0$};		%=======================================Group===========================================%
				\draw[dashed] (-1.15,-0.7) rectangle (1.15,-5.5);
				\draw[dashed] (1.7,-0.7) rectangle (4.1,-5.5);
				\draw[dashed] (4.6,-0.7) rectangle (7,-5.5);
				\draw[dashed] (7.5,-0.7) rectangle (9.85,-5.5);
				\draw[dashed] (10.4,-0.7) rectangle (12.8,-5.5);
				\draw[dashed] (13.3,-0.7) rectangle (15.65,-5.5);
				\draw[dashed] (16.2,-0.7) rectangle (18.0,-5.5);
				\draw[dashed] (18.6,-0.7) rectangle (20.4,-5.5);
				\draw[dashed] (21.0,-0.7) rectangle (22.8,-5.5);
				\draw[dashed] (23.4,-0.7) rectangle (25.2,-5.5);
			\end{tikzpicture}
			\captionsetup{justification=centering}
			\caption{The $t$-th PCG of the $(10,7)$ MDS code $\mathcal{Q}_3$, where $t\in [0:3)$ and nodes 4, 5 are chosen as goal nodes.}
			\label{fig the PCG of Q3}
		\end{figure}
		%====================================PCG of Q4=============================================%
		
		\begin{figure}[!htbp]
			\tikzset{
				source1/.style={draw,thick,rounded corners,align=center,inner sep=.1cm,minimum width=1.3cm,minimum height=0.8cm},
				source2/.style={draw,thick,rounded corners,align=center,inner sep=.1cm,minimum width=2.6cm,minimum height=0.8cm},
				global scale/.style={scale=#1,every node/.append style={scale=#1}}
			}
			\centering
			\begin{tikzpicture}[global scale=0.67]
				%=======================================node number=======================================================%
				\node [source2] (GN 0) at (0,0.2)  { PCG-data \\ of node 0};
				%\node at (1.3,0.1){$+$};
				\node [source2] (GN 2) at (3.0,0.2)  { PCG-data \\  of node 1};
				\node [source2] (RN 0) at (6,0.2)  { PCG-data \\  of node 2};
				\node [source2] (RN 0) at (9,0.2)  { PCG-data \\  of node 3};
				\node [source2] (RN 0) at (12,0.2)  { PCG-data \\  of node 4};
				\node [source2] (RN 0) at (15,0.2)  { PCG-data \\  of node 5};
				\node [source2] (RN 0) at (18,0.2)  { PCG-data \\  of node 6};
				\node [source2] (RN 0) at (21,0.2)  { PCG-data \\  of node 7};
				%\node  at (19.95,0.1)  {$=0$};
				%====================================PCG-data of first row===================================================%
				\node  at (0,-1)  {\small $f_{0,0}\hspace{-0.5mm}-\hspace{-0.5mm}16^t f_{0,1}$};
				\node  at (1.45,-1) {\small $+$};
				\node  at (3,-1)  {\small $15^tf_{1,0}$};
				\node  at (4.5,-1) {\small $+$};
				\node  at (6,-1)  {\small $3^tf_{2,0}\hspace{-0.5mm}-\hspace{-0.5mm}14^tf_{2,2}$};
				\node  at (7.5,-1) {\small $+$};
				\node  at (9,-1)  {\small $13^tf_{3,0}$};
				\node  at (10.5,-1) {\small $+$};
				\node  at (12,-1)  {\small $5^tf_{4,0}\hspace{-0.5mm}-\hspace{-0.5mm}12^tf_{4,4}$};
				\node  at (13.5,-1) {\small $+$};
				\node  at (15,-1)  {\small $11^tf_{5,0}$};
				\node  at (16.5,-1) {\small $+$};
				\node  at (18,-1)  {\small $7^tf_{6,0}\hspace{-0.5mm}-\hspace{-0.5mm}10^tf_{6,8}$};
				\node  at (19.5,-1)  {\small $+$};
				\node  at (21,-1)  {\small $9^tf_{7,0}$};
				\node  at (22.7,-1) {$=0$};
				%====================================PCG-data of second row===================================================%
				\node  at (0,-1.6)  {\small $16^tf_{0,1}$};
				\node  at (1.45,-1.6) {\small $+$};
				\node  at (3.0,-1.6)  {\small $2^tf_{1,1}\hspace{-0.5mm}-\hspace{-0.5mm}15^tf_{1,0}$};
				\node  at (4.5,-1.6) {\small $+$};
				\node  at (6.0,-1.6)  {\small $3^tf_{2,1}-14^tf_{2,3}$};
				\node  at (7.5,-1.6) {\small $+$};
				\node  at (9.0,-1.6)  {\small $13^tf_{3,1}$};
				\node  at (10.5,-1.6) {\small $+$};
				\node  at (12.0,-1.6) {\small $5^tf_{4,1}\hspace{-0.5mm}-\hspace{-0.5mm}12^tf_{4,5}$};
				\node  at (13.5,-1.6) {\small $+$};
				\node  at (15.0,-1.6)  {\small $11^tf_{5,1}$};
				\node  at (16.5,-1.6) {\small $+$};
				\node  at (18.0,-1.6)  {\small $7^tf_{6,1}\hspace{-0.5mm}-\hspace{-0.5mm}10^tf_{6,9}$};
				\node  at (19.5,-1.6)  {\small $+$};
				\node  at (21.0,-1.6)  {\small $9^tf_{7,1}$};
				\node  at (22.7,-1.6)  {$=0$};
				%====================================PCG-data of third row===================================================%
				\node  at (0,-2.2)  {\small $ f_{0,2}\hspace{-0.5mm}-\hspace{-0.5mm}16^t f_{0,3}$};
				\node  at (1.45,-2.2) {\small $+$};
				\node  at (3.0,-2.2)  {\small $15^tf_{1,2}$};
				\node  at (4.5,-2.2) {\small $+$};
				\node  at (6.0,-2.2)  {\small $14^tf_{2,2}$};
				\node  at (7.5,-2.2) {\small $+$};
				\node  at (9.0,-2.2)  {\small $4^tf_{3,2}\hspace{-0.5mm}-\hspace{-0.5mm}13^tf_{3,0}$};
				\node  at (10.5,-2.2) {\small $+$};
				\node  at (12.0,-2.2)  {\small $5^tf_{4,2}\hspace{-0.5mm}-\hspace{-0.5mm}12^tf_{4,6}$};
				\node  at (13.5,-2.2) {\small $+$};
				\node  at (15.0,-2.2)  {$11^tf_{5,2}$};
				\node  at (16.5,-2.2) {\small $+$};
				\node  at (18.0,-2.2)  {\small $7^tf_{6,2}\hspace{-0.5mm}-\hspace{-0.5mm}10^tf_{6,10}$};
				\node  at (19.5,-2.2)  {\small $+$};
				\node  at (21.0,-2.2)  {\small $9^tf_{7,2}$};
				\node  at (22.7,-2.2)  {$=0$};
				%====================================PCG-data of fourth row===================================================%
				\node  at (0,-2.8)  {\small $16^t f_{0,3}$};
				\node  at (1.45,-2.8) {\small $+$};
				\node  at (3.0,-2.8)  {\small $2^tf_{1,3}\hspace{-0.5mm}-\hspace{-0.5mm}15^tf_{1,2}$};
				\node  at (4.5,-2.8) {\small $+$};
				\node  at (6.0,-2.8)  {\small $14^tf_{2,3}$};
				\node  at (7.5,-2.8) {\small $+$};
				\node  at (9.0,-2.8)  {\small $4^tf_{3,3}\hspace{-0.5mm}-\hspace{-0.5mm}13^tf_{3,1}$};
				\node  at (10.5,-2.8) {\small $+$};
				\node  at (12.0,-2.8) {\small $5^tf_{4,3}\hspace{-0.5mm}-\hspace{-0.5mm}12^tf_{4,7}$};
				\node  at (13.5,-2.8) {\small $+$};
				\node  at (15.0,-2.8)  {\small $11^tf_{5,3}$};
				\node  at (16.5,-2.8) {\small $+$};
				\node  at (18.0,-2.8)  {\small $7^tf_{6,3}\hspace{-0.5mm}-\hspace{-0.5mm}10^tf_{6,11}$};
				\node  at (19.5,-2.8)  {\small $+$};
				\node  at (21.0,-2.8)  {\small $9^tf_{7,3}$};
				\node  at (22.7,-2.8)  {$=0$};
				%====================================PCG-data of fifth row===================================================%
				\node  at (0,-3.4)  {\small $f_{0,4}-16^t f_{0,5}$};
				\node  at (1.45,-3.4) {\small $+$};
				\node  at (3.0,-3.4)  {\small $15^tf_{1,4}$};
				\node  at (4.5,-3.4) {\small $+$};
				\node  at (6.0,-3.4)  {\small $3^tf_{2,4}\hspace{-0.5mm}-\hspace{-0.5mm}14^tf_{2,6}$};
				\node  at (7.5,-3.4) {\small $+$};
				\node  at (9.0,-3.4)  {\small $13^tf_{3,4}$};
				\node  at (10.5,-3.4) {\small $+$};
				\node  at (12.0,-3.4)  {\small $12^tf_{4,4}$};
				\node  at (13.5,-3.4) {\small $+$};
				\node  at (15.0,-3.4)  {\small $6^tf_{5,4}\hspace{-0.5mm}-\hspace{-0.5mm}11^tf_{5,0}$};
				\node  at (16.5,-3.4) {\small $+$};
				\node  at (18.0,-3.4)  {\small $7^tf_{6,4}\hspace{-0.5mm}-\hspace{-0.5mm}10^tf_{6,12}$};
				\node  at (19.5,-3.4)  {\small $+$};
				\node  at (21.0,-3.4)  {\small $9^tf_{7,4}$};
				\node  at (22.7,-3.4)  {$=0$};
				%====================================PCG-data of sixth row===================================================%
				\node  at (0,-4.0)  {\small $16^tf_{0,5}$};
				\node  at (1.45,-4.0) {\small $+$};
				\node  at (3.0,-4.0)  {\small $2^tf_{1,5}\hspace{-0.5mm}-\hspace{-0.5mm}15^tf_{1,4}$};
				\node  at (4.5,-4.0) {\small $+$};
				\node  at (6.0,-4.0)  {\small $3^tf_{2,5}\hspace{-0.5mm}-\hspace{-0.5mm}14^tf_{2,7}$};
				\node  at (7.5,-4.0) {\small $+$};
				\node  at (9.0,-4.0)  {\small $13^tf_{3,5}$};
				\node  at (10.5,-4.0) {\small $+$};
				\node  at (12.0,-4.0) {\small $12^tf_{4,5}$};
				\node  at (13.5,-4.0) {\small $+$};
				\node  at (15.0,-4.0)  {\small $6^tf_{5,5}\hspace{-0.5mm}-\hspace{-0.5mm}11^tf_{5,1}$};
				\node  at (16.5,-4.0) {\small $+$};
				\node  at (18.0,-4.0)  {\small $7^tf_{6,5}\hspace{-0.5mm}-\hspace{-0.5mm}10^tf_{6,13}$};
				\node  at (19.5,-4.0)  {\small $+$};
				\node  at (21.0,-4.0)  {$9^tf_{7,5}$};
				\node  at (22.7,-4.0)  {$=0$};
				%====================================PCG-data of seventh row===================================================%
				\node  at (0,-4.6)  {\small $f_{0,6}-16^t f_{0,7}$};
				\node  at (1.45,-4.6) {\small $+$};
				\node  at (3,-4.6)  {\small $15^tf_{1,6}$};
				\node  at (4.5,-4.6) {\small $+$};
				\node  at (6,-4.6)  {\small $14^tf_{2,6}$};
				\node  at (7.5,-4.6) {\small $+$};
				\node  at (9,-4.6)  {\small $4^tf_{3,6}\hspace{-0.5mm}-\hspace{-0.5mm}13^tf_{3,4}$};
				\node  at (10.5,-4.6) {\small $+$};
				\node  at (12,-4.6)  {$12^tf_{4,6}$};
				\node  at (13.5,-4.6) {\small $+$};
				\node  at (15,-4.6)  {\small $6^tf_{5,6}\hspace{-0.5mm}-\hspace{-0.5mm}11^tf_{5,2}$};
				\node  at (16.5,-4.6) {\small $+$};
				\node  at (18,-4.6)  {\small $7^tf_{6,6}\hspace{-0.5mm}-\hspace{-0.5mm}10^tf_{6,14}$};
				\node  at (19.5,-4.6)  {\small $+$};
				\node  at (21.0,-4.6)  {\small $9^tf_{7,6}$};
				\node  at (22.7,-4.6) {$=0$};
				%====================================PCG-data of eighth row===================================================%
				\node  at (0,-5.2)  {$16^t f_{0,7}$};
				\node  at (1.45,-5.2) {\small $+$};
				\node  at (3,-5.2)  {\small $2^tf_{1,7}\hspace{-0.5mm}-\hspace{-0.5mm}15^tf_{1,6}$};
				\node  at (4.5,-5.2) {\small $+$};
				\node  at (6,-5.2)  {$14^tf_{2,7}$};
				\node  at (7.5,-5.2) {\small $+$};
				\node  at (9,-5.2)  {\small $4^tf_{3,7}\hspace{-0.5mm}-\hspace{-0.5mm}13^tf_{3,5}$};
				\node  at (10.5,-5.2) {\small $+$};
				\node  at (12.0,-5.2) {$12^tf_{4,7}$};
				\node  at (13.5,-5.2) {\small $+$};
				\node  at (15,-5.2)  {\small $6^tf_{5,7}\hspace{-0.5mm}-\hspace{-0.5mm}11^tf_{5,3}$};
				\node  at (16.5,-5.2) {\small $+$};
				\node  at (18.0,-5.2)  {\small $7^tf_{6,7}\hspace{-0.5mm}-\hspace{-0.5mm}10^tf_{6,15}$};
				\node  at (19.5,-5.2)  {\small $+$};
				\node  at (21.0,-5.2)  {\small $9^tf_{7,7}$};
				\node  at (22.7,-5.2)  {$=0$};
				%====================================PCG-data of ninth row===================================================%
				\node  at (0,-5.8)  {\small $f_{0,8}\hspace{-0.5mm}-\hspace{-0.5mm}16^t f_{0,9}$};
				\node  at (1.45,-5.8) {\small $+$};
				\node  at (3,-5.8)  {$15^tf_{1,8}$};
				\node  at (4.5,-5.8) {\small $+$};
				\node  at (6,-5.8)  {\small $3^tf_{2,8}\hspace{-0.5mm}-\hspace{-0.5mm}14^tf_{2,10}$};
				\node  at (7.5,-5.8) {\small $+$};
				\node  at (9,-5.8)  {\small $13^tf_{3,8}$};
				\node  at (10.5,-5.8) {\small $+$};
				\node  at (12,-5.8)  {\small $5^tf_{4,8}\hspace{-0.5mm}-\hspace{-0.5mm}12^tf_{4,12}$};
				\node  at (13.5,-5.8) {\small $+$};
				\node  at (15,-5.8)  {\small $11^tf_{5,8}$};
				\node  at (16.5,-5.8) {\small $+$};
				\node  at (18.0,-5.8)  {\small $10^tf_{6,8}$};
				\node  at (19.5,-5.8)  {\small $+$};
				\node  at (21.0,-5.8)  {\small $8^tf_{7,8}\hspace{-0.5mm}-\hspace{-0.5mm}9^tf_{7,0}$};
				\node  at (22.7,-5.8)   {$=0$};
				%====================================PCG-data of second row===================================================%
				\node  at (0,-6.4)  {\small $16^tf_{0,9}$};
				\node  at (1.45,-6.4) {\small $+$};
				\node  at (3.0,-6.4)  {\small $2^tf_{1,9}-15^tf_{1,8}$};
				\node  at (4.5,-6.4) {\small $+$};
				\node  at (6,-6.4)  {\small $3^tf_{2,9}-14^tf_{2,11}$};
				\node  at (7.5,-6.4) {\small $+$};
				\node  at (9,-6.4)  {\small $13^tf_{3,9}$};
				\node  at (10.5,-6.4) {\small $+$};
				\node  at (12,-6.4) {\small $5^tf_{4,9}-12^tf_{4,13}$};
				\node  at (13.5,-6.4) {\small $+$};
				\node  at (15,-6.4)  {\small $11^tf_{5,9}$};
				\node  at (16.5,-6.4) {\small $+$};
				\node  at (18.0,-6.4)  {\small $10^tf_{6,9}$};
				\node  at (19.5,-6.4)  {\small $+$};
				\node  at (21.0,-6.4)  {\small $8^tf_{7,9}\hspace{-0.5mm}-\hspace{-0.5mm}9^tf_{7,1}$};
				\node  at (22.7,-6.4)  {$=0$};
				%====================================PCG-data of third row===================================================%
				\node  at (0,-7.0)    {\small $f_{0,10}\hspace{-0.5mm}-\hspace{-0.5mm}16^t f_{0,11}$};
				\node  at (1.45,-7.0) {\small $+$};
				\node  at (3,-7.0)  {\small $15^tf_{1,10}$};
				\node  at (4.5,-7.0) {\small $+$};
				\node  at (6,-7.0)  {\small $14^tf_{2,10}$};
				\node  at (7.5,-7.0) {\small $+$};
				\node  at (9,-7.0)  {\small $4^tf_{3,10}\hspace{-0.5mm}-\hspace{-0.5mm}13^tf_{3,8}$};
				\node  at (10.5,-7.0) {\small $+$};
				\node  at (12,-7.0)  {\small $5^tf_{4,10}\hspace{-0.5mm}-\hspace{-0.5mm}12^tf_{4,14}$};
				\node  at (13.5,-7.0) {\small $+$};
				\node  at (15,-7.0)  {\small $11^tf_{5,10}$};
				\node  at (16.5,-7.0) {\small $+$};
				\node  at (18,-7.0)  {\small $10^tf_{6,10}$};
				\node  at (19.5,-7.0)  {\small $+$};
				\node  at (21,-7.0)  {\small $8^tf_{7,10}\hspace{-0.5mm}-\hspace{-0.5mm}9^tf_{7,2}$};
				\node  at (22.7,-7.0)  {$=0$};
				%====================================PCG-data of fourth row===================================================%
				\node  at (0,-7.6)  {\small $16^t f_{0,11}$};
				\node  at (1.45,-7.6) {\small $+$};
				\node  at (3,-7.6)  {\small $2^tf_{1,11}\hspace{-0.5mm}-\hspace{-0.5mm}15^tf_{1,10}$};
				\node  at (4.5,-7.6) {\small $+$};
				\node  at (6,-7.6)  {$14^tf_{2,11}$};
				\node  at (7.5,-7.6) {\small $+$};
				\node  at (9,-7.6)  {\small $4^tf_{3,11}\hspace{-0.5mm}-\hspace{-0.5mm}13^tf_{3,9}$};
				\node  at (10.5,-7.6) {\small $+$};
				\node  at (12,-7.6) {\small $5^tf_{4,11}\hspace{-0.5mm}-\hspace{-0.5mm}12^tf_{4,15}$};
				\node  at (13.5,-7.6) {\small $+$};
				\node  at (15,-7.6)  {\small $11^tf_{5,11}$};
				\node  at (16.5,-7.6) {\small $+$};
				\node  at (18,-7.6)  {\small $10^tf_{6,11}$};
				\node  at (19.5,-7.6)  {\small $+$};
				\node  at (21,-7.6)  {\small $8^tf_{7,11}\hspace{-0.5mm}-\hspace{-0.5mm}9^tf_{7,3}$};
				\node  at (22.7,-7.6) {$=0$};
				%====================================PCG-data of fifth row===================================================%
				\node  at (0,-8.2)  {\small $f_{0,12}-16^t f_{0,13}$};
				\node  at (1.45,-8.2) {\small $+$};
				\node  at (3,-8.2)  {\small $15^tf_{1,12}$};
				\node  at (4.5,-8.2) {\small $+$};
				\node  at (6,-8.2)  {\small $3^tf_{2,12}\hspace{-0.5mm}-\hspace{-0.5mm}14^tf_{2,14}$};
				\node  at (7.5,-8.2) {\small $+$};
				\node  at (9,-8.2)  {$13^tf_{3,12}$};
				\node  at (10.5,-8.2) {\small $+$};
				\node  at (12,-8.2)  {$12^tf_{4,12}$};
				\node  at (13.5,-8.2) {\small $+$};
				\node  at (15,-8.2)  {\small $6^tf_{5,12}\hspace{-0.5mm}-\hspace{-0.5mm}11^tf_{5,8}$};
				\node  at (16.5,-8.2) {\small $+$};
				\node  at (18,-8.2)  {$10^tf_{6,12}$};
				\node  at (19.5,-8.2)  {\small $+$};
				\node  at (21,-8.2)  {\small $8^tf_{7,12}\hspace{-0.5mm}-\hspace{-0.5mm}9^tf_{7,4}$};
				\node  at (22.7,-8.2)  {$=0$};
				%====================================PCG-data of sixth row===================================================%
				\node  at (0,-8.8)  {$16^tf_{0,13}$};
				\node  at (1.45,-8.8) {\small $+$};
				\node  at (3,-8.8)  {\small $2^tf_{1,13}\hspace{-0.5mm}-\hspace{-0.5mm}15^tf_{1,12}$};
				\node  at (4.5,-8.8) {\small $+$};
				\node  at (6,-8.8)  {\small $3^tf_{2,13}\hspace{-0.5mm}-\hspace{-0.5mm}14^tf_{2,15}$};
				\node  at (7.5,-8.8) {\small $+$};
				\node  at (9,-8.8)  {\small $13^tf_{3,13}$};
				\node  at (10.5,-8.8) {\small $+$};
				\node  at (12,-8.8) {\small $12^tf_{4,13}$};
				\node  at (13.5,-8.8) {\small $+$};
				\node  at (15,-8.8)  {\small $6^tf_{5,13}\hspace{-0.5mm}-\hspace{-0.5mm}11^tf_{5,9}$};
				\node  at (16.5,-8.8) {\small $+$};
				\node  at (18,-8.8)  {$10^tf_{6,13}$};
				\node  at (19.5,-8.8)  {\small $+$};
				\node  at (21,-8.8)  {\small $8^tf_{7,13}\hspace{-0.5mm}-\hspace{-0.5mm}9^tf_{7,5}$};
				\node  at (22.7,-8.8)  {$=0$};
				%====================================PCG-data of seventh row===================================================%
				\node  at (0,-9.4)  {\small $f_{0,14}-16^t f_{0,15}$};
				\node  at (1.45,-9.4) {\small $+$};
				\node  at (3,-9.4)  {\small $15^tf_{1,14}$};
				\node  at (4.5,-9.4) {\small $+$};
				\node  at (6,-9.4)  {\small $14^tf_{2,14}$};
				\node  at (7.5,-9.4) {\small $+$};
				\node  at (9,-9.4)  {\small $4^tf_{3,14}\hspace{-0.5mm}-\hspace{-0.5mm}13^tf_{3,12}$};
				\node  at (10.5,-9.4) {\small $+$};
				\node  at (12,-9.4)  {\small $12^tf_{4,14}$};
				\node  at (13.5,-9.4) {\small $+$};
				\node  at (15,-9.4)  {\small $6^tf_{5,14}\hspace{-0.5mm}-\hspace{-0.5mm}11^tf_{5,10}$};
				\node  at (16.5,-9.4) {\small $+$};
				\node  at (18,-9.4)  {\small $10^tf_{6,14}$};
				\node  at (19.5,-9.4)  {\small $+$};
				\node  at (21,-9.4)  {\small $8^tf_{7,14}\hspace{-0.5mm}-\hspace{-0.5mm}9^tf_{7,6}$};
				\node  at (22.7,-9.4)  {$=0$};
				%====================================PCG-data of eighth row===================================================%
				\node  at (0,-10)  {\small $16^t f_{0,15}$};
				\node  at (1.5,-10) {\small $+$};
				\node  at (3,-10)  {\small $2^tf_{1,15}\hspace{-0.5mm}-\hspace{-0.5mm}15^tf_{1,14}$};
				\node  at (4.5,-10) {\small $+$};
				\node  at (6,-10)  {\small $14^tf_{2,15}$};
				\node  at (7.5,-10) {\small $+$};
				\node  at (9,-10)  {\small $4^tf_{3,15}\hspace{-0.5mm}-\hspace{-0.5mm}13^tf_{3,13}$};
				\node  at (10.5,-10) {\small $+$};
				\node  at (12,-10) {\small $12^tf_{4,15}$};
				\node  at (13.5,-10) {\small $+$};
				\node  at (15,-10)  {\small $6^tf_{5,15}\hspace{-0.5mm}-\hspace{-0.5mm}11^tf_{5,11}$};
				\node  at (16.5,-10) {\small $+$};
				\node  at (18,-10)  {$10^tf_{6,15}$};
				\node  at (19.5,-10)  {\small $+$};
				\node  at (21,-10)  {\small $8^tf_{7,15}\hspace{-0.5mm}-\hspace{-0.5mm}9^tf_{7,7}$};
				\node  at (22.7,-10) {$=0$};
				%=======================================Group===========================================%
				\draw[dashed] (-1.25,-0.7) rectangle (1.15,-10.3);
				\draw[dashed] (1.7,-0.7) rectangle (4.25,-10.3);
				\draw[dashed] (4.7,-0.7) rectangle (7.25,-10.3);
				\draw[dashed] (7.7,-0.7) rectangle (10.25,-10.3);
				\draw[dashed] (10.7,-0.7) rectangle (13.25,-10.3);
				\draw[dashed] (13.7,-0.7) rectangle (16.25,-10.3);
				\draw[dashed] (16.7,-0.7) rectangle (19.25,-10.3);
				\draw[dashed] (19.7,-0.7) rectangle (22.25,-10.3);
			\end{tikzpicture}
			\captionsetup{justification=centering}
			\caption{The $t$-th PCG of the desired $(8,5)$ MDS code $\mathcal{Q}_4$ over $\mathbb{F}_{17}$ with sub-packetization level 16, where $t\in [0:3)$ and nodes 6, 7 are chosen as goal nodes.}
			\label{fig the PCG of Q4}
		\end{figure}
		
		Furthermore, Table \ref{table the data and LSE} gives the indices of data downloaded from each helper node and the indices of equations chosen from the $t$-th PCG of code $\mathcal{Q}_4$ when repairing a failed node.
		\begin{table}[htbp]
			\centering
			\caption{The data downloaded and equations used to repair a failed node of code $\mathcal{Q}_4$}\label{table the data and LSE}
			\begin{tabular}{|c|c|c|}
				\hline \multirow{2}{*}{Failed nodes} & The indices of data downloaded  & The rows of Figure \ref{fig the PCG of Q4} used\\
				& from each helper node & to repair the failed node\\
				\hline $0$ & $\{0,2,4,6,8,10,12,14\}$ & $\{1,3,5,7,9,11,13,15\}$  \\
				\hline $1$ & $\{1,3,5,7,9,11,13,15\}$ & $\{2,4,6,8,10,12,14,16\}$  \\
				\hline $2$ & $\{0,1,4,5,8,9,12,13\}$ & $\{1,2,5,6,9,10,13,14\}$  \\
				\hline $3$ & $\{2,3,6,7,10,11,14,15\}$ & $\{3,4,7,8,11,12,15,16\}$  \\
				\hline $4$ & $\{0,1,2,3,8,9,10,11\}$ & $\{1,2,3,4,9,10,11,12\}$  \\
				\hline $5$ & $\{4,5,6,7,12,13,14,15\}$ & $\{5,6,7,8,13,14,15,16\}$  \\
				\hline $6$ & $\{0,1,2,3,4,5,6,7\}$ & $\{1,2,3,4,5,6,7,8\}$  \\
				\hline $7$ & $\{8,9,10,11,12,13,14,15\}$ & $\{9,10,11,12,13,14,15,16\}$  \\
				\hline
			\end{tabular}
		\end{table}
		
	\end{Example}

	\section{Conclusion}
	In this paper, we provided a generic transformation that can generate new MDS code with $\delta$-optimal access property for an arbitrary set of two nodes from any existing MDS code, where $2\le \delta\le r$. A new explicit construction of high-rate MDS code $\mathcal{C}$ was obtained by directly applying the transformation to a scalar MDS code, whose sub-packetization level is much smaller than that of the Ye-Barg codes 1 and 2. Moreover, the new code $\mathcal{C}$ can support any single value $\delta \in [2:r]$, which is much more flexible than the Vajha-Babu-Kumar code and the Sasidharan-Myna-Kumar code.  Transformations that can generate new MDS code with $\delta$-optimal access property for a set of more than two nodes and support multiple $\delta$ simultaneously will be part of our on-going work.

\end{document}